\documentclass[12pt, draftclsnofoot, onecolumn]{IEEEtran}
\input{psfig.sty}
\input{epsf.sty}
 \usepackage{hyperref} 
\usepackage{fancybox}   
\usepackage{psfig}      
\usepackage{graphicx}
\usepackage{amsmath}
\usepackage{color}
\usepackage{epsfig}
\usepackage{amssymb}
\usepackage{verbatim}
\usepackage{enumitem}
\usepackage{bbm}
\usepackage{lipsum}

\IEEEoverridecommandlockouts

\newcommand{\remove}[1]{}

\newtheorem{theorem}{Theorem}

\newtheorem{proposition}{Proposition}

\newtheorem{remark}{Remark}

\newtheorem{algorithm}{Algorithm}
\newtheorem{assumption}{Assumption}

\usepackage{setspace}
\usepackage{subfigure}

\begin{document}

 \title{Gibbsian On-Line Distributed Content Caching Strategy for Cellular Networks 
 \thanks{Email: achattop@usc.edu, Bartek.Blaszczyszyn@ens.fr, keeler@wias-berlin.de}
 \thanks{This work was done when Arpan Chattopadhyay was working in INRIA, Paris as a postdoc.}
}

\author{
Arpan~Chattopadhyay, Bart{\l}omiej~B{\l}aszczyszyn, H. Paul Keeler \\
}

\maketitle
\thispagestyle{empty}

\begin{abstract}
In this paper, we develop Gibbs sampling based techniques for learning the optimal placement of contents in a cellular network. 
We consider the situation where a {\em finite collection} of base stations are scattered on the plane, each covering a   
cell (possibly overlapping with other cells). Mobile users request for downloads from a finite set of contents according 
to some popularity distribution which may be known or unknown to the base stations. Each base station has a fixed memory 
space that can store only a strict subset of the contents at a time; hence, if a user requests for a content 
that is not stored at any of its serving base stations, the content has to be downloaded from the backhaul. Hence, we consider  
the problem of optimal content placement which minimizes the rate of download from the backhaul, or equivalently 
maximize the cache hit rate. It is known that, when multiple cells 
can overlap with one another (e.g., under dense deployment of base stations in small cell networks), it is not optimal 
to place the most popular contents in each base station. However, the optimal content placement problem is NP-complete. 
Using  ideas of Gibbs sampling, we propose simple sequential content update rules that decide whether to 
store a content 
at a base station (if required from the base station) and which content has to be removed from the corresponding cache, 
based on the knowledge of contents stored in its neighbouring 
base stations. 
The update rule is shown to be asymptotically converging 
to the optimal content placement for all nodes under the knowledge of content popularity. Next, we extend the algorithm to 
address the situation where content popularities and cell topology are initially unknown, but are estimated as new requests arrive to the base 
stations; we show that our algorithm working with the running estimates of content popularities and cell topology also converges asymptotically 
to the optimal content placement. Finally, we   demonstrate  the improvement in cache hit rate 
compared to most popular content placement and independent content placement strategies via numerical exploration.
\end{abstract}


\vspace{-5mm}
\section{Introduction}\label{section:introduction}
The proliferation of smartphones and tablets equipped with 3G and 4G connectivity 
and the fast growing demand for downloading multimedia files 
have resulted in severe overload in the internet backhaul, 
and it is expected to be worse with the advent of 5G in near future. 
Recent idea of densifying 
cellular networks will improve wireless throughput, but this will 
eventually push the backhaul bandwidth to its limit. In order to alleviate 
this problem, the idea of caching popular multimedia contents has recently been proposed. 
Given the fact that the popular contents 
are requested many times which results in network congestion, one way to 
reduce the congestion is to cache the popular contents 
at various intermediate nodes in the network. In case of cellular network, 
this requires adding physical memory to base stations (BSs): 
macro, micro, nano and pico. This has several advantages: 
(i) Caching contents at base stations reduce backhaul load. 
(ii) Caching 
reduces delay in fetching the content, thereby reducing the multimedia playback time. 
(iii) Caching will allow the end user to 
download a lower quality content in case his channel quality or bad 
or in case he wants to control his total amount of download.

Under dense placement of base stations, it is often the case that the {\em cells} 
(a cell is defined to be a region around a BS where 
the user is able to get sufficient downlink data rate from the BS) of different 
BSs might overlap with each other in an arbitrary manner 
(see \cite{bartek-keeler15sinr-process-poisson-networks-factorial-moment-measures}). 
Hence, if a user is covered by multiple BSs, she has the option to download a content 
from any one of the serving BSs. This gives rise to the 
problem of optimal content placement in the caches of cellular BSs (see \cite{bartek14caching}, \cite{femtocaching}); 
the trade-off is that ideally the caching strategy should avoid placing the same 
content  in two BSs whose cells have a significant overlap, 
while it is not desirable for the non-overlapped region.\footnote{However, this claim does not hold when multiple base stations having their own caches cooperate not only at the cache level but also at the signal level; see \cite{xu-etal16cooperative-tx} and 
\cite{zhen-etal16cooperative-caching}. We consider only cache-level multi-cell cooperation in our current paper.} 
{\em Optimal content placement under such situation requires global knowledge of base station locations and cell topologies, 
and solving the optimization problem requires intensive computation.  In order to tackle these problems, we develop 
sequential cache update algorithms motivated by Gibbs sampling, that asymptotically lead to optimal content placement, where each base station updates 
its contents only when a new content is downloaded from the backhaul to meet a user request, and this update is 
made solely based on the knowledge of the neighbouring BSs whose cells have nonzero intersection with the 
cell of the BS under consideration. The results are also extended to the case where the content popularities and cell topology 
are unknown initially and 
are learnt over time as new content requests arrive to the base stations.}  
Numerical results demonstrate the improvement in cache hit rate using Gibbs sampling technique for cache update, 
compared to most popular content placement and independent content placement strategies in the caches.

\vspace{-5mm}
\subsection{Related Work}\label{subsection:related_work}
There have been considerable amount of work in the literature dedicated to cellular caching. Benefits and challenges 
for caching in 5G networks have been described in \cite{wang14caching}. The authors of \cite{martina-14caching} 
have developed a method to analyze the  performance of caches (isolated or networked), and shown that placing the most 
popular subset of contents in each cache is not optimal in case of interconnected caches. The paper 
\cite{femtocaching} deals with optimal content placement in wireless caches given BS-user association. 
The authors of \cite{poularakis13exploiting} have addressed the problem of optimal content placement under user mobility. 
The authors of \cite{bartek14caching} have proposed a randomized content placement scheme in cellular BS caches 
in order to maximize cache hit rate, but their 
scheme assumes that the  contents are placed independently across the caches, which is obviously suboptimal. This work was later extended 
to the case of heterogeneous networks in \cite{serbetci16caching}. The authors of \cite{liu16caching} 
have again considered independent probabilistic caching in a random heterogeneous network. 
The paper \cite{avrachenkov16caching} has addressed the problem of cache miss minimization in a random network setting. 
The authors of \cite{debbah2016caching} 
have studied the problem of distributed caching in
ultra-dense wireless small cell networks using mean field games; however, 
this formulation requires us to take base station density to infinity (which may not be 
true in practice), and it does not provide any guarantee on the optimality of this caching strategy. 
The paper \cite{naveen-postdoc-paper} proposes a pricing based scheme for jointly assigning content requests to cellular BSs and 
updating the cellular caches; but this paper focuses on certain cost minimization instead of hit rate maximization, 
and it is optimal only when we can represent the data by very large number of chunks which can be used in employing rateless code. The problem of collaborative but decentralized caching among small base stations for a certain cost minimization has been analyzed in \cite{pantisano-etal14in-network-caching}, under the assumption that the caches have access to the contents of other caches connected to the same gateway; their formulation involves a certain cost for retrieval of a content from another cache. The authors of \cite{chiang-liao16encore-energy-aware-caching} 
address the problem of minimizing energy consumption under multi-cell transmission cooperation for interference reduction and content caching in heterogeneous networks. Since content providers might have to pay cellular network operators for caching their contents, an important question is how to cache contents among multiple base stations to that the content placement charge is minimized; this problem has been addressed by the work reported in \cite{gharaibeh-etal16efficient-online-collaborative-caching}. 
The authors of \cite{ostovari-etal16collaborative-caching} have considered the problem of collaborative content  placement at caches of multiple base stations, but under the assumption that cache sizes at base stations are unlimited. The paper \cite{jiang-etalYYoptimal-cooperative-content-caching}  discusses cooperative content caching and delivery policy among multiple base stations. The paper \cite{avrachenkov2017low} provides a fast but suboptimal solution based on potential game formulation, to the problem of minimizing cache miss rate when multiple base stations have overlapping cells. \cite{avrachenkov2017low} also provided one simulated annealing-based algorithm (different from our Gibbs sampling approach) that minimizes the cache miss rate.

The paper \cite{bastug-etal14cache-enabled} analyzes a stochastic geometry framework where cache-enabled small base stations are randomly placed on infinite two dimensional plane, and calculated the expressions for the outage probability of a typical user (jointly in terms of SINR and content availability at the cache), as well as the delivery rate. The authors of \cite{bastug-etal16edge-caching}, for a randomly 
deployed heterogeneous network, derive approximate expressions for the average delivery rate considering inter-tier and intra-tier dependence.  The authors of \cite{bastug-etal16geographical-caching} 
analyze the average delay of users for a random two-tier network under perfect knowledge of content popularity distribution and randomized caching policies. All of these papers consider a stochastic geometry framework for base station locations, and assume limited backhaul. However, in our current paper, we consider known placement of a finite number of base stations,  and seek to maximize the cache hit rate over the entire network; {\em thus, our work seeks to reduce the load in the backhaul without imposing a hard constraint on the backhaul capacity}. It is worth mentioning that, under this setting, we provide decentralized cache update schemes which are hit-rate optimal for a finite network in a time-average sense.

The authors of \cite{bharath16learning-caching} and \cite{leconte16placing-dynamic-content-caches} 
propose learning schemes for unknown time-varying 
popularity of contents, but their scheme does not have theoretical guarantee of convergence 
to the optimal content placement {\em across the network when cells of different BSs overlap with each other}. 
The paper \cite{moharir14high-dimensional} establishes that, when popularity is dynamic, any scheme that separates 
content popularity estimation and cache update (i.e., control) phases is strictly order-wise suboptimal in terms of hit rate. A big data approach has been taken in \cite{bastug-etal15big-data-caching} for estimating content popularities empirically from mobile traffic data collected from a telecom operator. The authors of \cite{neglia2017cache} proposed simulated annealing based caching for a single cache, and also addressed the issue of unknown content popularities by proposing an algorithm that avoids direct popularity estimation. 

Contrary to the prior literature, our current paper provides theoretical guarantee of  convergence for an optimal distributed  cellular 
cache update scheme   that maximizes the time-average cache hit rate  over the network involving  caches in multiple  base stations with overlapping cells; this minimizes the amount of data downloaded from the backhaul. The results also hold when popularities and cell topology are unknown initially and are learnt over time using 
the information of request arrivals in the base stations.

\vspace{-5mm}

\subsection{Organization and Our Contribution}\label{subsection:our_contribution}
The rest of the paper is organized as follows. 

\begin{itemize}
\item The system model has been defined in Section~\ref{section:system-model}.
\item In Section~\ref{section:given-temperature}, we propose an update scheme for the caches based on the knowledge of the 
contents cached in neighbouring BSs. The update scheme is based on Gibbs sampling techniques, and cache updates are 
made only when new content requests arrive. The scheme asymptotically converges to a near-optimal content placement 
in the network, since the scheme is proposed for a finite ``inverse temperature'' to be defined later. We prove convergence 
of the proposed scheme. To the best of our knowledge, such a scheme has never been 
used in the context of  caching in cellular network.
\item In Section~\ref{section:varying-inverse-temperature}, we discuss how to slowly increase the inverse  temperature to 
$\infty$ so that the near-optimal limiting solution in Section~\ref{section:given-temperature} actually 
converges to the globally optimal solution. We provide rigorous proof for the convergence of this scheme.
\item In Section~\ref{section:learning-popularities}, we discuss how to adapt the update schemes to the situation 
when unknown content popularities and cell topology are learnt over time as new content requests arrive to the BSs over time.
\item In Section~\ref{section:numerical}, we numerically demonstrate that the proposed Gibbs sampling approach has the potential 
to significantly improve the cache hit rate in cellular networks.
\item Finally, we conclude in Section~\ref{section:conclusion}.
\end{itemize}

\vspace{-5mm}

\section{System Model and Notation}\label{section:system-model}
\vspace{-5mm}

\subsection{Network Model}\label{subsection:network-model}
We consider a {\em finite} set $\mathcal{N}:=\{1,2,\cdots,N \}$ of  base stations (BSs) on the two-dimensional Euclidean space. The location of the base stations are {\em deterministic and arbitrary}; for example, the locations could come from a given realization of a point process over a finite geographical region.  
The set of points covered by a BS 
constitute the {\em cell} of the corresponding BS. This coverage could be signal to noise ratio (SNR) based coverage 
where a point is covered 
by a BS if and only if the SNR at that point from the BS exceeds some threshold. 
We denote 
the cell of BS~$i$ ($1 \leq i \leq N$) by $\mathcal{C}_i$. Let us define $\mathcal{C}:=\cup_{i=1}^N \mathcal{C}_i$.   
The area of any subset $\mathcal{A}$ of $\mathbb{R}^2$ is denoted 
by $|\mathcal{A}|$. We allow the cells of various BSs to have arbitrary and different {\em finite} areas. 
The cells of two BSs might have a nonzero intersection; any downlink 
mobile user located at such an intersection is covered by more than one BS. 
Let us denote by $2^{\mathcal{N}}$ the collection of all subsets of $\mathcal{N}$, and let $s$ denote one such generic subset. 
Let us denote by 
$\mathcal{C}(s):=(\cap_{i \in s} \mathcal{C}_i ) \cap ( \cup_{i \notin s} \mathcal{C}_i )^{c}$ the region in $\mathcal{C}$ 
which is covered only by the BSs from the subset $s$. 
See Figure~\ref{fig:caching-cell-diagram} for a better understanding of the cell model.

\begin{figure}[!t]
\begin{center}
\includegraphics[height=4cm, width=5cm]{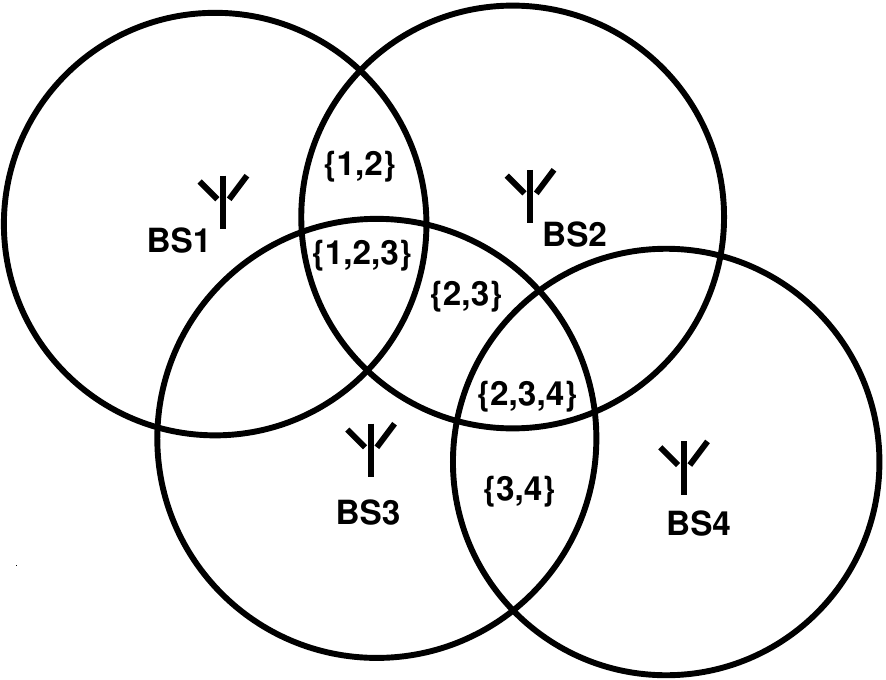}
\end{center}
\caption{A pictorial description of the base station  coverage model. In the diagram, four BSs are shown with numbers 
$1,2,3,4$. The circles correspond to the cells of the BSs. 
The region marked as $\{1,2,3\}$ has $s=\{1,2,3\}$; i.e., this region is called $\mathcal{C}(\{1,2,3\})$ 
and it is covered only by BSs $\{1,2,3\}$ and no other BS. Similar meaning 
applies to other regions.}
\label{fig:caching-cell-diagram}
\vspace{-5mm}
\end{figure}

\vspace{-5mm}

\subsection{Content Request Process}\label{subsection:content-request}
Contents from a set $\mathcal{M}:=\{1,2,\cdots,M\}$ are requested by users located inside $\mathcal{C}$. {\em We assume that each of these 
contents have the same size, though we will explain at the end of Section~\ref{section:given-temperature} how to easily take care of unequal content sizes in our analysis. }
Content~$i$ ($1 \leq  i \leq M$) is  
requested by users according to a homogeneous Poisson point process 
in space (inside $\mathcal{C}$) and time with intensity $\lambda_i$; this is the expected 
number of requests for content~$i$ per second per square meter inside $\mathcal{C}$. Let $\lambda:=\sum_{i=1}^M \lambda_i$. Note that, 
$\frac{\lambda_i}{\lambda}$ denotes the probability that a content request is for content~$i$; in other words, 
$\frac{\lambda_i}{\lambda}$ is the popularity of content~$i$. We also assume that $\lambda_1 \geq \lambda_2 \geq \cdots \geq \lambda_M$.

It is worth mentioning that the cache request process  essentially follows the popularly known  independent request model (IRM) as described in \cite{paschos-etal16wireless-caching}. \footnote{In case  content popularities are time-varying (e.g., under the shot noise model as described in \cite{paschos-etal16wireless-caching}), our proposed scheme in Section~$\ref{section:learning-popularities}$ for cache update while learning content popularities by observing the content request arrival process, will work fine so long as the content popularites change at a rate slow enough so that the popularity estimates converge to the actual popularity  between two successive changes in the content popularity.} In our model, 
the total request arrival process to the system is a time-homogeneous Poisson process with intensity  $\lambda |\mathcal{C}|$. The only difference with IRM is that unlike IRM, the popularity of content~$i$ in our model, $\frac{\lambda_i}{\lambda}$, can be arbitrary and do not necessarily follow a power law; this makes our request arrival model more general. As an example, this model is valid when users are scattered on the two-dimensional plane according to a homogeneous Poisson point process, and each user is generating content requests according to a time-homogeneous Poisson process, and any new request is for content~$i$ with probability $\frac{\lambda_i}{\lambda}$.

\vspace{-5mm}

\subsection{Content Caching at BSs}\label{subsection:content-caching}
We assume that each BS can store $K$ number of contents, where $K < M$. Let $B$ denote a generic configuration of content placement in caches of the 
network. $B$ is defined as a $M \times N$ matrix with $B_{i,j}=1$ if content~$i$ is stored at the cache of BS~$j$, and 
$B_{i,j}=0$ otherwise. Note that, any feasible $B$ must satisfy $\sum_{i=1}^M B_{i,j}=K$ for all $j \in \{1,2,\cdots,N \}$; we rule out the possibility 
of $\sum_{i=1}^M B_{i,j}<K$ since that will be a waste of cache memory resources in BSs. 
Let us denote the set of all feasible configurations by $\mathcal{B}$. 
Clearly, the cardinality of $\mathcal{B}$ is ${{M}\choose{K}}^N$. 
Apart from $B$, we will also use the symbol $A$ for a generic configuration belonging to set $\mathcal{B}$.

\vspace{-5mm}

\subsection{Cache Hit Rate Maximization Problem}\label{subsection:content-requests-hit-rate}
We assume that, whenever a new request for a content arrives, it is served 
by one  BS covering that point and having the content in its cache; if a content request is served from the cache, we call  
the event as a {\em cache hit}. In case no covering BS has the content (i.e., no cache hit, or {\em cache miss}), 
the content needs to be downloaded by one of the covering BSs and served to the user 
(this will be explained later). The requests do not tolerate any delay; i.e., we do not consider 
the possibility of holding the requests in a queue and serving the content to users in batch 
once the content becomes available in a BS. Also, we assume infinite bandwidth available for all downlink transmissions; i.e., 
each content is assumed to be served instantaneously.\footnote{This is a valid assumption when the downlink traffic in the network is light. Even under heavy traffic, in a small cell network, downlink capacity is typically high and the number of users per cell is small. On the other hand, a backhaul link typically serves many base stations. This makes the backhaul capacity a major bottleneck. While  a joint optimization of throughput or delay considering the availability of contents in the caches and considering instantaneous backhaul traffic load and downlink traffic load is highly useful, the problem becomes too hard in general when there are a lot of small cells over a large geographical region. Hence, we decide to decouple the caching problem and downlink load management problem at each cell, which is equivalent to infinite downlink bandwidth assumption for the caching problem. It is important to note that, even with this simplification, there remains significant challenge in the problem of content assignment to the caches.

Limited downlink capacity was considered in prior work such as \cite{bastug-etal14cache-enabled} and \cite{bastug-etal16edge-caching}, but they did not propose optimal caching strategy for base stations with overlapping cells.}

Let the random variable $H_B$ denote the number of cache hits in the entire network in unit time, under configuration $B$. 
We define the cache hit rate 
$h(B)=\mathbf{E}(H_B)$ where the expectation is over the randomness in the content request arrival process.  Clearly, 
\begin{eqnarray}
 h(B)=\sum_{s \in 2^{\mathcal{N}}} 
 |\mathcal{C}(s)| \sum_{i=1}^M \lambda_i  \mathbf{1}\{\sum_{j \in s} B_{i,j} \geq 1\}. 
 \label{eqn:expanded-expression-of-hit-rate}
\end{eqnarray}

In this paper, we are interested in finding an optimal configuration which achieves: 
\begin{equation}
\sup_{B \in \mathcal{B}}h(B). \label{eqn:objective-function}
\end{equation}

Cache hit rate has been considered as the objective function in prior literature; see \cite{avrachenkov16caching}, \cite{bartek14caching}, \cite{martina-14caching} and \cite{sengupta-etal14learning-distributed-caching} for reference. The authors of \cite{sengupta-etal14learning-distributed-caching} considered  hit rate maximization under coded caching. However, one can consider other objective functions such as latency in content delivery as in \cite{tandon-simone16cloud-aided-edge-caching}; we choose cache hit rate  since it is a commonly used objective function. Cache hit rate is a suitable objective function when the requested content needs to be served instantaneously; if the requests are delay-tolerant, then queueing of the requests and contents are allowed and there latency in content delivery would be a more suitable objective function.  It is worth mentioning that, in case requests and contents are allowed to be queued at the base station, there is no formal proof that maximizing hit rate will minimize the latency in content delivery, though intuitively one can expect so.

\eqref{eqn:objective-function} is  an optimization problem with $0-1$ integer variables, 
nonlinear objective function 
and linear constraints. This class of problems has been shown to be NP-complete (see \cite{karp-complexity-paper}), 
and hence, we cannot expect any 
polynomial time algorithm to solve \eqref{eqn:objective-function}. Hence, in this section, we provide iterative, distributed 
cache update scheme that asymptotically solves the problem. However, since the algorithm is iterative, we cannot use the optimal 
configuration over infinite time horizon. Hence, we seek to design a randomized iterative cache update scheme which yields 
\begin{equation}\label{eqn:asymptotic-target}
 \liminf_{T \rightarrow \infty} \frac{\int_0^T \mathbf{E}(h(R(\tau))) d \tau}{T} = \sup_{B \in \mathcal{B}}h(B).
\end{equation}
where $R(\tau) \in \mathcal{B}$ is the configuration of all caches in the network at time $\tau$. Our iterative scheme is randomized, which renders 
$R(\tau)$ a random variable; hence, we work with the expectation $\mathbf{E}$.

It is important to note that, by  maximizing the cache hit rate, we seek to minimize the download rate from the backhaul; this is necessary because backhaul capacity is limited in practice, and,  also, downloading a content from a server via the backhaul link might involve certain cost. However, we do not consider any specific upper limit on the backhaul link capacity. If the backhaul link is blocked due to heavy load or due to finite backhaul capacity, a content request which is not able to find a match in the caches of its covering base stations can either be dropped or kept waiting for service hoping that  the backhaul load will be reduced later. If the content request arrival statistics is approximately known to the network operator prior to cache installation at the base stations, the operator can simulate the cache update scheme and estimate the average download rate required for the backhaul under the scheme; this estimate can be used as a design guideline for choosing the backhaul capacity.  
Hence, for the rest of the paper, we assume sufficient backhaul capacity to deal with cache miss.

\vspace{-5mm}

\section{Cache Update via Basic Gibbs Sampling}\label{section:given-temperature}
In this section, we propose an iterative, randomized cache update scheme so that the time-average occupancy of each $B \in \mathcal{B}$ under the scheme follows certain distribution called Gibbs distribution. In Section~\ref{section:varying-inverse-temperature}, we explain how tuning a certain parameter of the Gibbs distribution helps us in solving problem~\eqref{eqn:asymptotic-target}.

Let us rewrite \eqref{eqn:expanded-expression-of-hit-rate} as $h(B)=\sum_{j=1}^N h_j(B)$ where 
\begin{eqnarray}
 h_j(B)= \sum_{i=1}^M  \lambda_i  \sum_{s \in 2^{\mathcal{N}} } 
\frac{ |\mathcal{C}(s)| B_{i,j} \mathbf{1}\{j \in s\}  }{ \max\{1, \sum_{k \in s} B_{i,k} \} }.
 \label{eqn:expanded-expression-of-per-node-hit-rate}
\end{eqnarray}
We call $ h_j(B)$ to be the  cache hit rate seen by BS~$j$ under configuration~$B$. This will be the true cache 
hit rate seen by BS~$j$ under configuration~$B$ if a new content request is served by one covering BS chosen uniformly from the set 
of covering BSs having that content. Note that, if more than one covering BSs have that content, choice of the serving BS will not affect 
the hit rate; hence, we can safely assume uniform choosing of the serving BS.

In order to solve $\sup_{B \in \mathcal{B}} \sum_{j=1}^N h_j(B)$, we propose to employ Gibbs sampling techniques 
(see \cite[Chapter~$7$]{breamud99gibbs-sampling}). Let us assume that each BS maintains a {\em virtual} cache capable of 
storing $K$ contents. The broad idea is that one can update the virtual cache contents in an iterative fashion 
using Gibbs sampling. Whenever a content is requested from a BS not having the content in its 
physical (real) cache, the BS will download it from the backhaul and, at 
the same time, will decide  to store it in the real cache depending on whether it is stored in its virtual cache or not.

We will update the virtual caches according to a stochastic iterative algorithm so that the steady state probability 
of configuration $B$ becomes: 
\begin{equation}\label{eqn:expression-for-gibbs-stationary-probability}
\pi_{\beta}(B):=\frac{e^{\beta h(B)}}{ \sum_{B^{'} \in \mathcal{B}} e^{\beta h(B^{'})} }:= \frac{e^{\beta h(B)}}{Z_{\beta}} ,
\end{equation}
where $\beta$ is called the ``inverse temperature'' (motivated by literature from statistical Physics), and 
$Z_{\beta}$ is called the {\em partition function}.

Note that, $\lim_{\beta \rightarrow \infty} \sum_{B \in \arg \max_{B^{'} \in \mathcal{B}} h(B^{'})} \frac{e^{\beta h(B)}}{ \sum_{B^{'} \in \mathcal{B}} e^{\beta h(B^{'})} }=1$. 
Hence, if we choose configuration $B$ for all virtual caches with probability 
$\pi_{\beta}(B)$, then, for sufficiently large $\beta$, the chosen configuration will belong to  
$\arg \max_{B \in \mathcal{B}} h(B)$ with probability close to $1$. If real cache configuration closely follows virtual cache configuration, 
we can achieve near-optimal cache hit rate for real caching system.

\vspace{-5mm}

\subsection{Gibbs sampling approach for ``virtual'' cache update}
\label{subsection:modified-Gibbs-given-temperature}
Let us consider discrete time instants $t=0, 1,2 , \cdots$ when virtual cache contents are updated; this is different from the continuous 
time $\tau$ used before. Let us denote the 
configuration in all virtual caches in the network after the $t$-th decision instant by $V(t)$, where $V(t) \in \mathcal{B}$.   
The Gibbs sampling algorithm simulates a discrete-time Markov chain $V(t)$ on state space 
$\mathcal{B}$, whose stationary probability 
distribution is given by 
$\pi_{\beta}(B)= \frac{e^{\beta h(B)}}{Z_{\beta}}$.

Let us define the set of {\em neighbours} of BS~$j$ ({\bf including BS~$j$}) as 
$\Psi(j):=\{n: n \in \mathcal{N}, \mathcal{C}_j \cap \mathcal{C}_n \neq \emptyset \}$. Let us denote by $B_{\cdot,-j}$ the restriction 
of configuration $B$ to all BSs except BS~$j$, i.e., $B_{\cdot,-j}$ is obtained by deleting the $j$-th column of $B$. Let 
$\pi_{\beta}(\cdot | B_{\cdot,-j})$ denote the conditional distribution of the network-wide configuration conditioned on $B_{\cdot,-j}$, 
under the joint distribution $\pi_{\beta}(\cdot)$. Clearly, $\pi_{\beta}(A | B_{\cdot,-j})=0$ if $A_{\cdot,-j} \neq B_{\cdot,-j}$. 

If $A_{\cdot,-j} = B_{\cdot,-j}$, then 
\begin{equation}\label{eqn:first-expression-for-conditional-probability}
\pi_{\beta}(A | B_{\cdot,-j})=\frac{e^{\beta h( A )}}{\sum_{v_j \in \{0,1\}^M, ||v_j||_1=K} e^{\beta h( v_j,B_{\cdot,-j} )}},
\end{equation}
where $||v_j||_1$ is the sum of all components of the vector $v_j$.

Note that, there is common factor $ e^{\beta \sum_{n \notin \Psi(j)} h_n (B) }$ in both numerator and denominator 
of the expression in \eqref{eqn:first-expression-for-conditional-probability}, since this term does 
not depend on the contents in the virtual cache in BS~$j$. Hence, \eqref{eqn:first-expression-for-conditional-probability} 
can be further simplified as:
\begin{equation}\label{eqn:second-expression-for-conditional-probability}
\pi_{\beta}(A | B_{\cdot,-j})=\frac{e^{\beta \sum_{n \in \Psi(j)}h_n(A)}}{\sum_{v_j \in \{0,1\}^M, ||v_j||_1=K} e^{\beta \sum_{n \in \Psi(j)}h_n(v_j,B_{\cdot,-j} ) }}.
\end{equation}

Now, let us define $h_n(A,s)$ to be the hit rate seen by BS~$n$ 
under configuration $A$ due to the content requests generated  from the region $\mathcal{C}(s)$. Clearly, 
$h_n(A)=\sum_{s \in 2^{\mathcal{N}}} h_n(A,s)$, since the hit rate at BS~$n$ under configuration $A$ is equal to the sum of hit rates 
by requests generated from all possible segments $\{\mathcal{C}(s)\}_{s \in 2^{\mathcal{N}}}$. Now, note that, the term 
$e^{\beta \sum_{n \in \Psi(j)} \sum_{s: j \notin s} h_n(A,s)}$ is a common factor in the numerator and denominator of the expression 
in \eqref{eqn:second-expression-for-conditional-probability}, since 
this factor does not depend on the contents in the virtual cache of BS~$j$. Hence, when 
$A_{\cdot,-j} = B_{\cdot,-j}$, we can simplify \eqref{eqn:second-expression-for-conditional-probability} further as follows:

\small
\begin{equation}\label{eqn:third-expression-for-conditional-probability}
\pi_{\beta}(A | B_{\cdot,-j})=\frac{e^{\beta \sum_{n \in \Psi(j), s \ni j }h_n(A,s)}}{\sum_{v_j \in \{0,1\}^M, ||v_j||_1=K} e^{\beta \sum_{n \in \Psi(j), s \ni j}h_n(v_j,B_{\cdot,-j}, s ) }},
\end{equation}
\normalsize

\noindent
where 
\begin{equation}\label{eqn:hnA_definition}
h_n( A,s )=  \sum_{i=1}^M  \lambda_i  
 \frac{ |\mathcal{C}(s)| A_{i,n} \mathbf{1}\{n \in s\}  }{ \max\{1, \sum_{k \in s} A_{i,k} \} }. 
\end{equation}

We now describe an algorithm for sequentially updating the network-wide virtual cache configuration $V(t)$.

\vspace{2mm}
\noindent\fbox{
    \parbox{\textwidth}{
\begin{algorithm}\label{algorithm:virtual-cache-update-basic-gibbs-sampling}
 Start with an arbitrary $V(0) \in \mathcal{B}$. 
 At discrete time $t$, pick a node $j_t \in \mathcal{N}$ randomly having uniform distribution from $\mathcal{N}$. 
 Then, update the contents in the virtual cache of BS~$j_t$ by picking up a network-wide virtual cache configuration $A \in \mathcal{B}$ 
 with probability $\pi_{\beta}(A | V_{\cdot,-j_t}(t-1))$. Only contents in the virtual cache of BS~$j_t$ are modified  
 by this operation.
\end{algorithm}
}}
\vspace{2mm}

\begin{proposition}
 Under Algorithm~\ref{algorithm:virtual-cache-update-basic-gibbs-sampling}, 
 $\{V(t)\}_{t \geq 0}$ is a reversible  Markov chain, and it achieves the steady-state 
probability distribution $\pi_{\beta}(B)= \frac{e^{\beta h(B)}}{Z_{\beta}}$. 
\end{proposition}
\begin{proof}
The proof is standard, and it follows from the theory in \cite[Chapter~$7$]{breamud99gibbs-sampling}).
\end{proof}

\begin{remark}
 In Algorithm~\ref{algorithm:virtual-cache-update-basic-gibbs-sampling}, in 
 order to make an update at time $t$, BS~$j_t$ needs to know the contents of the virtual caches only 
 from  $\Psi(j_t)$. This requires information exchange between BS~$j_t$ and its neighbours in each slot. Such information exchange 
 may happen through the backhaul network, but this does not exert much load on the backhaul since the actual contents 
 are not exchanged via the backhaul in this process.
 \end{remark}

 \begin{remark}\label{remark:complexity-issue-in-basic-Gibbs}
   The denominator in the simplified sampling probability expression in \eqref{eqn:third-expression-for-conditional-probability}  
 requires a summation over all possible virtual cache configurations in $\Psi(j_t)$. This allows the system 
 to avoid the huge combinatorial problem of calculating $Z_{\beta}$ which requires 
 ${{M}\choose{K}}^N$ addition operations. The advantage will be even more visible if we consider 
 the possibility of varying $\beta$ with time or learning $\{\lambda_i\}_{1 \leq i \leq M}$ over time if they are not 
 known; the optimization problem $\sup_{B \in \mathcal{B}} h(B)$ will change over time in this case, and it will require 
 calculation of the partition function in each slot. However, for large $M$ and $K$, the $O({{M} \choose {K}})$ computations per iteration  in \eqref{eqn:third-expression-for-conditional-probability} can still be large; in this case, at each $t$, one can randomly remove one content from the virtual cache of $j_t$ and then replace it by one content (from $(M-K+1)$ contents not present in the virtual cache of $j_t$) using Gibbs sampling; this will involve a summation in the denominator of \eqref{eqn:third-expression-for-conditional-probability} over all $(M-K+1)$ possible configurations that can possibly result from this replacement, and  it will require only $O(M-K+1)$ computations. One can easily show that $V(t)$ will be a reversible Markov chain with stationary distribution $\pi_{\beta}(\cdot)$ under this variant. However, for the sake of notational simplicity, we do not consider this variant in the theory part of the paper.
 \end{remark}

\vspace{-5mm}

\subsection{The real cache update scheme for fixed $\beta$}
\label{subsection:real-cache-update}
Now we propose a cache update scheme for the real caches present in the BSs. {\em Our scheme decides to store 
a content in the cache of a BS only when the content is requested from that BS}. This eliminates the necessity of any 
unnecessary download from the backhaul.

Let us consider content request arrivals at continuous time (denoted by $\tau$ again) to the BS. Let us recall that the 
virtual caches are updated only at discrete times $t=0,1,2,\cdots$. We assume that these discrete time instants 
$t=0,1,2,\cdots$ units are superimposed on the continuous time axis $\tau \geq 0$. Hence, $V(\tau)$ is defined to be equal to 
$V(t)$ for $\tau \in [t,t+1)$, where $t \in \mathbb{Z}_{+}$. 

Let us consider an increasing sequence of positive real numbers (viewed as time durations) 
 $T_1,T_2, T_3,\cdots$ such as $T_k \uparrow \infty$ as $k \rightarrow \infty$. Let $S_l:=T_1+T_2+\cdots+T_l$. 
 Let $\kappa(\tau):=\sup \{l \in \mathbb{Z}_{+}: S_l \leq \tau \} $ and 
 $\zeta(\tau):=S_{\kappa(\tau)}$. 

The real cache update scheme is given as follows:

\vspace{2mm}
\noindent\fbox{
    \parbox{\textwidth}{
\begin{algorithm}\label{algorithm:real-cache-update-algorithm}
 Start with some arbitrary $R(0) \in \mathcal{B}$. 
 
 At time $\tau$, if the 
 request for content~$i$ arrives at BS~$j$ (either because no other covering BS has this content or because 
 BS~$j$ has been chosen  from among the covering BSs having content~$i$), then BS~$j$ does the following:  
 \begin{itemize}
 \item If BS~$j$ has content~$i$, it will serve that. 
 
 \item If BS~$j$ does not have content~$i$, it serves the content by downloading from the backhaul. Then content~$i$  
 is stored in the real cache of BS~$j$ if and only if $V_{i,j}(\zeta(\tau)-)=1$ (i.e., if content~$i$ was stored in the virtual cache 
 of BS~$j$ at time $\zeta(\tau)-$). If the BS~$j$  
 decides to store content~$i$ then, 
 in order to make room for the newly stored content~$i$, any content~$k \neq i$ 
 such that $V_{k,j}(\zeta(\tau)-)=0$ and $R_{k,j}(\tau)=1$, is removed 
 from the real cache of BS~$j$.   
 
 \end{itemize}
\end{algorithm}
}}
\vspace{2mm}

\begin{remark}
 The idea behind taking $T_k \rightarrow \infty$ as $k \rightarrow \infty$ in Algorithm~\ref{algorithm:real-cache-update-algorithm} 
 is as follows. We know that $V(t)$ reaches the distribution 
 $\pi_{\beta}(\cdot)$ as $t \rightarrow \infty$. As $k \rightarrow \infty$, the fraction of time spent during $\tau \in [S_k, S_{k+1})$
 in copying the contents present in $V(S_k-)$ to real 
 caches becomes negligible, and the real caches are allowed to operate larger and larger fraction of time under  
 content distribution close to $\pi_{\beta}(\cdot)$.
\end{remark}

Now we make the following assumption:

\begin{assumption}\label{assumption:each-cell-has-a-region-covered-only-by-itself}
$| \mathcal{C}_i  \cap ( \cup_{j \neq i} \mathcal{C}_j )^{c}|>0$ for all $i \in \{1,2,\cdots,N\}$. 
\end{assumption}

\begin{theorem}\label{theorem:asymptotic-optimality-of-real-cache-update}
 Under Assumption~\ref{assumption:each-cell-has-a-region-covered-only-by-itself}, 
 Algorithm~\ref{algorithm:virtual-cache-update-basic-gibbs-sampling} and 
 Algorithm~\ref{algorithm:real-cache-update-algorithm}, we have (for the real caches in all BSs):
 $$\lim_{T \rightarrow \infty}\frac{  \int_0^T \mathbf{P}(R(\tau)=B) d \tau  }{T}=
  \frac{e^{\beta h(B)}}{ \sum_{B^{'} \in \mathcal{B}} e^{\beta h(B^{'})} }.$$
\end{theorem}
\begin{proof}
See Appendix~\ref{appendix:real-cache-update-basic-gibbs-proofs}.
 \end{proof}

\begin{remark}
 Note that, Assumption~\ref{assumption:each-cell-has-a-region-covered-only-by-itself} is very crucial in the proof of 
 Theorem~\ref{theorem:asymptotic-optimality-of-real-cache-update}, because this assumption ensures that every BS 
 gets content requests at some nonzero arrival rate, and hence can update its real cache at strictly positive rate. If 
 Assumption~\ref{assumption:each-cell-has-a-region-covered-only-by-itself} is not satisfied, then one can still achieve  near optimal hit rate 
 in real caches. It is achieved under a scheme where  a new content request is sent to any of  its covering BSs with very small probability $\eta>0$, 
 and 
 otherwise the request is sent to a covering BS having that content. Similar analysis as in this paper can show that the time-average 
 expected hit rate under this scheme differs from the optimal hit rate $\max_{B \in \mathcal{B}} h(B)$ only by a small margin which goes to $0$ as 
 $\eta \downarrow 0$.
\end{remark}

\begin{remark}
Note that, Algorithm~\ref{algorithm:real-cache-update-algorithm} will work for any sequence 
$\{T_k\}_{k \geq 1}$ so long as the sequence increases to infinity. However, the speed of convergence will depend on the specific choice of the sequence, and also on system parameters such as content popularities, arrival rates and cellular network topology. An analytical characterization of the speed of convergence as a function of $\{T_k\}_{k \geq 1}$ is hard, so we leave it for future research endeavours on this topic.
\end{remark}

{\bf Incorporating unequal content sizes in our model:} If $c_i$ is the size of content~$i$ in bytes and $K$ is the memory of a cache in bytes, then any feasible configuration $B$ for unequal content sizes must satisfy the condition $\sum_{i=1}^M B_{i,j} c_i \leq K$ for all $j \in \{1,2,\cdots,N\}$ (instead of $\sum_{i=1}^M B_{i,j}  \leq K$ for all $j \in \{1,2,\cdots,N\}$ as required for equal content sizes with $K$ being the maximum possible number of contents per cache); the collection of such feasible $B$ matrices is called $\mathcal{B}$. Clearly, the set of feasible configurations $\mathcal{B}$  is redefined   for unequal content sizes. However, given this new $\mathcal{B}$, Algorithm~\ref{algorithm:real-cache-update-algorithm} will still work since the virtual and real cache update schemes depend on the set $\mathcal{B}$ (which is a collection of $0-1$  matrices) and not on the actual content sizes. Convergence of all algorithms proposed later will also hold in case content sizes are unequal, though the convergence rates will vary depending on the exact $\mathcal{B}$. Note that,   for unequal content sizes, the best choice  of $h(B)$ is the  mean cache hit rate in bytes per second, i.e., $h(B):=\sum_{s \in 2^{\mathcal{N}}} 
 |\mathcal{C}(s)| \sum_{i=1}^M \lambda_i c_i  \mathbf{1}\{\sum_{j \in s} B_{i,j} \geq 1\} $. This new objective function is separable across base stations and hence the virtual cache update rules for fixed $\beta$ will have similar form as \eqref{eqn:second-expression-for-conditional-probability} and \eqref{eqn:third-expression-for-conditional-probability}; as a result, this modified $h(B)$ will not alter the structures of the algorithms at all. 
 {\em  For the rest of the paper, we will use \eqref{eqn:expanded-expression-of-hit-rate} as a definition of $h(B)$ for the sake of simplicity.}


\vspace{-5mm}

\section{Varying $\beta$ to Reach Optimality}\label{section:varying-inverse-temperature}
In this section, we discuss how to vary the inverse temperature $\beta$ to infinity with time so that the 
Gibbs sampling algorithm (used to 
update virtual caches) 
converges to the optimizer of \eqref{eqn:objective-function}. Here the intuition is that, Gibbs sampling with increasing $\beta$, combined with 
Algorithm~\ref{algorithm:real-cache-update-algorithm}  for real cache update, will achieve 
optimal time-average expected cache hit rate for problem \eqref{eqn:asymptotic-target}.

Let us define $$\Delta:=\max_{B \in \mathcal{B}}h(B)-\min_{B^{'} \in \mathcal{B}}h(B^{'})>0.$$

\vspace{2mm}
\noindent\fbox{
    \parbox{\textwidth}{
\begin{algorithm}\label{algorithm:virtual-cache-update-varying-inverse-temperature}
 This algorithm is analogous to Algorithm~\ref{algorithm:virtual-cache-update-basic-gibbs-sampling} except that, 
  at discrete time 
 instant $t \geq 0$, we use $\beta_t:=\beta_0 \log (1+t)$ instead of fixed $\beta$, where 
 $0< \beta_0 < \infty$ is the initial inverse temperature satisfying $\beta_0 N \Delta<1$ and 
 $\beta_0 \max_{B \in \mathcal{B}}h(B)<1$.
\end{algorithm}
}}
\vspace{2mm}

\begin{theorem}\label{theorem:strong-ergodicity-varying-inverse-temperature}
 Under Algorithm~\ref{algorithm:virtual-cache-update-varying-inverse-temperature} for virtual cache update, 
 the discrete time non-homogeneous Markov chain $\{V(t)\}_{t \geq 0}$ 
 is strongly ergodic, and the limiting distribution $\pi_{v,\infty}$ satisfies:
  $$\pi_{v,\infty}(\arg \max_{B \in \mathcal{B}} h(B))=1.$$ 
\end{theorem}
\begin{proof}
See Appendix~\ref{appendix:proof-of-strong-ergodicity-for-varying-inverse-temperature}  for the proof. The definition of strong ergodicity can be found in Appendix~\ref{appendix:weak-and-strong-ergodicity}. We have used some results from \cite[Chapter~$6$]{breamud99gibbs-sampling} in the proof. \footnote{In this connection, we would like to mention that  \cite{hajek1988cooling} also provided similar results as \cite[Chapter~$6$]{breamud99gibbs-sampling} on simulated annealing with a cooling schedule.}.
 \end{proof}

\begin{theorem}\label{theorem:asymptotic-optimal-real-cache-update-varying-inverse-temperature}
 Under Assumption~\ref{assumption:each-cell-has-a-region-covered-only-by-itself}, 
 Algorithm~\ref{algorithm:virtual-cache-update-varying-inverse-temperature} for virtual cache update and 
 Algorithm~\ref{algorithm:real-cache-update-algorithm} for real cache update, we have:
 $$\lim_{T \rightarrow \infty}\frac{  \int_0^T \mathbf{P}(R(\tau)=\arg \max_{B \in \mathcal{B}} h(B)) d \tau  }{T}=1,$$ 
and hence,
  $$\lim_{T \rightarrow \infty}\frac{  \int_0^T \mathbf{E}(h(R(\tau))) d \tau  }{T}= \max_{B \in \mathcal{B}} h(B).$$ 
\end{theorem}
\begin{proof}
 The first part of the proof follows using similar arguments as in the proof of Theorem~\ref{theorem:asymptotic-optimality-of-real-cache-update}. 
 The second part follows from the first part using the fact that $\mathbf{E}(h(R(\tau)))=\sum_{B \in \mathcal{B}} \mathbf{P}(R(\tau)=B) h(B)$.
\end{proof}

\begin{remark}
From \cite[Figure~$3$]{bartek14caching}, we notice that independent placement of contents across BSs can significantly 
outperform the placement of $K$ most popular contents in each BS cache (for a Poisson distributed network). However, our proposed scheme yields the optimal 
hit rate for every realization of the location of BSs, so long as the number of BSs is finite. Hence, we can safely claim that our 
proposed scheme significantly outperforms the placement of $K$ most popular contents in each BS cache.
\end{remark}

\vspace{-5mm}

\subsection{Convergence rate of the virtual cache update scheme}\label{subsection:convergence-speed} 
While we are not aware of any closed-form bound on the convergence rate 
for Algorithm~\ref{algorithm:virtual-cache-update-varying-inverse-temperature}, by using 
\cite[Chapter~$6$, Theorem~$7.2$]{breamud99gibbs-sampling}, we can provide convergence rate guarantee for 
Algorithm~\ref{algorithm:virtual-cache-update-basic-gibbs-sampling}. Let us consider 
the Markov 
chain $\{Y(l)\}_{l \geq 0}$, where $Y(l):=V(lN)$,  evolving under 
Algorithm~\ref{algorithm:virtual-cache-update-basic-gibbs-sampling}, and let us denote the corresponding transition probability matrix (t.p.m.)   by $Q$. Let us denote the Dobrushin ergodic coefficient of $Q$ by $\delta(Q)$ (see the proof of Theorem~\ref{theorem:strong-ergodicity-varying-inverse-temperature} in Appendix~\ref{appendix:proof-of-strong-ergodicity-for-varying-inverse-temperature}). 

Let us  define 
\begin{eqnarray*}
\Delta_1:=\max_{j \in \mathcal{N}, B \in \mathcal{B}} \max_{v_j,w_j \in \{0,1\}^M, |v_j|_1=|w_j|_1=K} | \sum_{n \in \Psi(j), s \ni j} h_n(v_j,B_{.,-j},s) -\sum_{n \in \Psi(j), s \ni j} h_n(w_j,B_{.,-j},s)|.
\end{eqnarray*} 
Note that, for any $j \in \mathcal{N}$, the quantity $h_n(v_j,B_{.,-j},s)$ for $n \in \Psi(j), s \ni j$ does not depend on the contents in the caches of base stations outside $\Psi(j)$. 

Now, let us recall Equation~\eqref{eqn:third-expression-for-conditional-probability}. In a way similar to the proof of Theorem~\ref{theorem:strong-ergodicity-varying-inverse-temperature} in Appendix~\ref{appendix:proof-of-strong-ergodicity-for-varying-inverse-temperature}, we can show $\delta(Q) \leq  1-\bigg(\frac{e^{-\beta \Delta_1}}{N}\bigg)^N $. Then, by \cite[Chapter~$6$, Theorem~$7.2$]{breamud99gibbs-sampling}, 
the total variation distance between $\mu_l$ (i.e., the probability distribution of $Y(l)$) and the steady state distribution  $\pi_{\beta}(\cdot)$ is upper bounded as:

\small
$$ d_V(\mu_l,\pi_{\beta}) \leq d_V(\mu_0,\pi_{\beta}) (\delta(Q))^l \leq d_V(\mu_0,\pi_{\beta}) \bigg( 1-\bigg(\frac{e^{-\beta \Delta_1}}{N}\bigg)^N \bigg)^l .$$
\normalsize

We can prove similar results for the Markov chain $\{V(lN+k)\}_{l \geq 0}$ for any $k \in \{0,1,\cdots,N-1\}$. 
Clearly, the R.H.S. of the above equation increases with $\beta$. Hence, under 
Algorithm~\ref{algorithm:virtual-cache-update-varying-inverse-temperature}, 
we can expect slower convergence rate as time increases. It has to be noted that there is a trade-off 
between convergence rate and the accuracy 
of the virtual cache update scheme using Gibbs sampling; higher accuracy (by taking very large $\beta$) 
obviously requires longer time because of slow convergence rate. It also suggests that the rate of convergence decreases 
with $N$ (provided that other parameters such as $\beta_0$ and $\Delta_1$ are fixed). Note that, $\Delta_1<\Delta$ and the difference between these two terms is large for large $N$. Hence, this provides a reasonably tight bound on the convergence rate for large $N$.

\vspace{-5mm}

\section{Learning Content Popularities and Cell Topology}\label{section:learning-popularities}
In previous sections, we assumed that the content request arrival rates per unit area, $\lambda_1,\lambda_2,\cdots,\lambda_M$, and the 
areas $|\mathcal{C}(s)|, s \in 2^{\mathcal{N}}$ are known 
 to all BSs. But, in practice, these quantities may not be known apriori, and one has to estimate these quantities 
 over time as new content requests arrive to the system. In this section, we will extend 
 Algorithm~\ref{algorithm:virtual-cache-update-varying-inverse-temperature} to adapt to learning of these quantities.
 
At time slot $t$, the BS $j_t$ (uniformly chosen from the set of BSs) chooses its virtual contents in such a way that 
 the probability of choosing network-wide configuration $A$ at time $t, lN \leq t \leq lN+N-1$ is 
$\pi_{\beta}(A | V_{\cdot,-j_t}(t-1))$.

 Let us recall the expression for $h_n(A,s)$ from \eqref{eqn:hnA_definition}. 
 Clearly, if one can estimate $\lambda_i |\mathcal{C}(s)|$ for all possible $(i,s) \in \mathcal{M} \times 2^{\mathcal{N}}$, 
 then one can have an estimate of $h_n(A,s)$. This can be done by estimating the request arrival rate for content~$i$ from 
 the region $\mathcal{C}(s)$; this is easy to do because this is a time-homogeneous Poisson process with rate 
 $\lambda_i |\mathcal{C}(s)|$ request per unit time.

Let us assume that each BS~$k$ has an estimate $\hat{\theta}(k,i,s,t)$ for  $\lambda_i | \mathcal{C}(s) |$ in slot $t$. 
This can be done through continuous 
message exchange among the BSs which observe the content request arrival process over time.

Now we present the virtual cache update algorithm.

\vspace{2mm}
\noindent\fbox{
    \parbox{\textwidth}{
\begin{algorithm}\label{algorithm:virtual-cache-update-learning}
 This algorithm is same as Algorithm~\ref{algorithm:virtual-cache-update-varying-inverse-temperature} except that 
the estimate $\hat{\theta}(k,i,s,t)$ is used at slot $t$ by BS~$k$, instead of the actual value of $\lambda_i | \mathcal{C}(s) |$.
\end{algorithm}
}}
\vspace{2mm}

\begin{assumption}\label{assumption:estimates-converge}
 $\lim_{t \rightarrow \infty} \hat{\theta}(k,i,s,t)=\lambda_i | \mathcal{C}(s) |$ almost surely 
for all $k \in \mathcal{N}, i \in \mathcal{M}, s \in 2^{\mathcal{N}}$.
\end{assumption}

Assumption~\ref{assumption:estimates-converge} ensures that each BS~$k$ has an estimate of the total request arrival rate for content~$i$ in the segment $\mathcal{C}(s)$ of the plane, and this estimate converges to the true value $\lambda_i | \mathcal{C}(s) |$ as time progresses. This can simply be achieved if the number of arrivals for various contents at each $\mathcal{C}(s)$ are recorded in the system, and are communicated periodically to all base stations in the network. As time progresses, more requests come to each segment $\mathcal{C}(s)$ and the estimates become better and closer to their respective mean values.

\begin{assumption}\label{assumption:uniqueness-of-maximizer}
 $\arg \max_{B \in \mathcal{B}}h(B)$ is unique.
\end{assumption}

\begin{theorem}\label{theorem:convergence-virtual-cache-update-learning}
 Under  Assumption~\ref{assumption:estimates-converge}, Assumption~\ref{assumption:uniqueness-of-maximizer} and 
 Algorithm~\ref{algorithm:virtual-cache-update-learning} for virtual cache update, 
 the discrete time non-homogeneous Markov chain $\{V(t)\}_{t \geq 0}$ 
 is strongly ergodic, and the limiting distribution $\pi_{v,\infty}(\cdot)$ satisfies $\pi_{v,\infty}(\arg \max_{B \in \mathcal{B}} h(B))=1$.
\end{theorem}
\begin{proof}
 See Appendix~\ref{appendix:proof-of-strong-ergodicity-for-varying-inverse-temperature-learning}.
\end{proof}

\begin{remark}
Assumption~\ref{assumption:uniqueness-of-maximizer} is a technical requirement for Theorem~\ref{theorem:convergence-virtual-cache-update-learning}. The reason is that, when $\lim_{t \rightarrow \infty} \beta_t=\infty$, the limiting transition probability matrix $Q^*$ of the non-homogeneous Markov chain $V(t)$ is  ergodic if there is a single maximizer in $\arg \max_{B \in \mathcal{B}}h(B)$, otherwise the ergodicity cannot be guaranteed; ergodicity of $Q^*$ is a technical requirement in the proof of Theorem~\ref{theorem:convergence-virtual-cache-update-learning}. However, we considered Algorithm~\ref{algorithm:virtual-cache-update-varying-inverse-temperature} for virtual cache update in the statement of  Theorem~\ref{theorem:convergence-virtual-cache-update-learning}, since it uses increasing $\beta_t$. In practical applications, $\beta$ will be kept constant  at a large but finite value, and $Q^*$ will be irreducible, ergodic in that case even when there are more than one maximizers; hence, Algorithm~\ref{algorithm:virtual-cache-update-basic-gibbs-sampling} for virtual cache update along with popularity and topology learning, will return an optimal configuration with the same high probability even when there are more than one maximizer configurations. Also, uniqueness of the maximizer is a practical assumption since, due to the non-uniform cell structure over a large region, it is highly unlikely that two different configurations will have the same hit rate. 
\end{remark}

\begin{theorem}\label{theorem:asymptotic-optimal-real-cache-update-varying-inverse-temperature-and-learning}
 Under Assumption~\ref{assumption:each-cell-has-a-region-covered-only-by-itself}, Assumption~\ref{assumption:estimates-converge}, 
 Assumption~\ref{assumption:uniqueness-of-maximizer}, 
 Algorithm~\ref{algorithm:virtual-cache-update-learning} for virtual cache update and 
 Algorithm~\ref{algorithm:real-cache-update-algorithm} for real cache update, the conclusions of 
 Theorem~\ref{theorem:asymptotic-optimal-real-cache-update-varying-inverse-temperature} hold.
\end{theorem}
\begin{proof}
The proof is similar to that of Theorem~\ref{theorem:asymptotic-optimal-real-cache-update-varying-inverse-temperature}.
\end{proof}

\vspace{-5mm}

\section{Performance Improvement Using Gibbs Sampling}\label{section:numerical}
In this section  we discuss the performance of the
proposed {\em Gibbs sampling content placement} (GSCP), which is based on
 Algorithm~\ref{algorithm:virtual-cache-update-basic-gibbs-sampling}. 
We  compare it with two
popular reference solutions: (i) {\em most popular content
placement} (MPCP) in each BS, and (ii) {\em independent content placement} (ICP) as in
\cite{bartek14caching}. Let us recall that, this latter method involves 
supplying all BSs a common  distribution with which each of them has to randomly choose its cache contents; 
this distribution is calculated as a 
function of the content popularities and BS coverage probabilities, so as 
to maximize the average cache hit probability of a typical   request.

\subsection{Optimality}
The MPCP is hit rate optimal when there is no cell overlapping.

The ICP maximizes the conditional cache hit probability (given coverage)  
{\em averaged over all possible locations  of
  BS}  in the (infinite, stationary) model, assuming some given coverage probabilities
(which is the distribution of the number of BSs covering a 
typical point) and independent selection of cache contents at each base station. It outperforms (on average) the MPCP (which can be seen
as  independent content placement with some particular, non-optimized, deterministic
content distribution); see \cite{bartek14caching} for details. The gain with respect to the MPCP is bigger when
there is more cell overlapping in the model. 
Our GSCP maximizes the hit rate (for a finite network deployment region) {\em for any  given  placement of  BSs}.

\subsection{Asymptotic performance}\label{subsection:asymptotic-performance}
It might be interesting to compare first the asymptotic performance of the
three solutions under two extremal situations:

\paragraph{Little overlapping cells} By this we mean a network where the
overlapping of cells is negligible. A specific example would be a Poisson 
Boolean model for which the product of the intensity of BSs and
the mean area of a cell is small. An extremal non-overlapping model
is  the Voronoi or, more generally, any tessellation. 

It is easy to see that in this
regime  all three solutions MPCP, ICP and GSCP are equivalent; all will
tend to store the most popular content in all BS. Hence, the
conditional hit
probability of a typical request, given coverage, is equal to $\frac{\sum_{i=1}^K \lambda_i}{\lambda}$, and 
the cache hit rate per unit covered area becomes $\sum_{i=1}^K \lambda_i$.

\paragraph{Highly overlapping cells} By this we mean a network where the  number of stations covering the
typical location increases in some sense to infinity as it is the case,
e.g., for the Poisson Boolean model with the product of the intensity of BSs and
the mean area of a cell going large.

While MPCP
always offers the same conditional hit probability given coverage (equal to $\frac{\sum_{i=1}^K \lambda_i}{\lambda}$), 
it can be shown under mild conditions that ICP
and GSCP are again equivalent with this conditional hit probability tending
to $1$, thus significantly outperforming the MPCP.

Sparse network and very dense network scenarios are not of  practical
interest. Hence, 
we  provide some numerical examples to show potential
performance improvement of GSCP with respect to ICP and MPCP.  It is to be noted that these numerical examples are  provided only to demonstrate the potential for performance improvement via Gibbs sampling approach. 
Providing guarantees for 
the actual margin of performance improvement for a more realistic network topology (such as Poisson Boolean model for cells) is 
left for future research endeavours.

\subsection{Distributed nature} 
The MPCP is completely distributed, i.e., all BSs fill in their caches
independently, provided that they know the content popularity distribution. This
popularity can be locally estimated, as it is suggested in Section~\ref{section:learning-popularities}. 

The ICP is also distributed, provided that the specific model-optimal
distribution of the contents is fed to the BSs. This
distribution depends on the coverage probabilities, which can be
estimated only over the entire network; they  cannot be calculated locally.
Hence  the ICP requires a central authority for the calculation of the optimal 
content distribution. 

Our GSCP is distributed in the sense that each BS updates its cache using 
only local estimation and local information exchange.

\subsection{Numerical example of performance improvement via Gibbs sampling for various values of $\beta$}\label{subsection:caching-two-dimensional-numerical-example}

We consider six BSs placed inside the unit square  bounded by the lines $x=0,x=1,y=0,y=1$ on the $xy$ plane. 
There are four contents 
$\mathcal{M}=\{1,2,3,4\}$ with popularity vector $(0.3,0.25,0.24,0.21)$. Each BS can store at most two contents (i.e., $K=2$). 
Content requests are being generated over the unit square according to a time and space homogeneous 
Poisson point process with intensity $1$ requests per unit time per unit area.

We consider two possible scenarios for the cells of base stations:
\begin{itemize}
\item {Scenario~$1$:} We assume that the six cells are either square or rectangular in size, and together cover the entire unit square. The corners of the cells are given by $\{(0,0),(0,0.5),(0.5,0),(0.5,0.5)\}$, \\ $\{(0.5,0), (1,0), (0.5,0.5),(1,0.5)\}$, 
$\{(0.5,0.5),(1,0.5),(0.5,1),(1,1)\}$, \\ $\{(0.25,0.25),(0.75,0.25),(0.25,0.75),(0.75,0.75)\}$,  $\{(0,0.5),(0.25,0.5),(0,1),(0.25,1)\}$\\ and $\{(0,0.75),(0,1),(0.5,0.75),(0.5,1)\}$.
\item {Scenario~$2$:} The six base stations are placed uniformly and independently inside the unit square (random placement). The cell of a base station is a circular region centered at it and with radius~$0.35$~units. The placement  realization in this numerical example left $8.45 \%$ area of the unit square uncovered by base stations; this area does not contribute to the cache hit rate. The location of the six base stations for this particular realization are $(0.7215, 0.8286  )$, $(0.3155, 0.8455  )$, $(0.7401, 0.0172 )$, $(0.0821,0.1970 )$, $(0.4580, 0.7946 )$ and $(0.5078, 0.0669)$.
\end{itemize}
For both scenarios, under most popular content placement, the cache hit rate is $(0.3+0.25)=0.55$ multiplied by the fraction of area covered by the base stations (this fraction is $1$ for scenario~$1$ but less than $1$ for scenario~$2$). 

For  scenario~$1$, we have also considered the case where all BSs choose the contents independently with the same probability distribution tuned to maximize the expected hit rate;  the expected hit rate turned out to be $0.6081$ in this case.

If the 
contents in all caches are chosen probabilistically according to the steady state Gibbs distribution $\pi_{\beta}(\cdot)$, 
one can expect that the expected cache hit rate improves as $\beta$ increases, and converges to the maximum possible cache 
hit rate as $\beta \uparrow \infty$. 

The above phenomena for scenario~$1$ and scenario~$2$  have been captured in 
Figure~\ref{fig:caching-plot-six-base-stations}. This figure also shows that even with finite but large $\beta$, significantly higher cache hit rate can be achieved 
asymptotically compared to the most popular content placement strategy for all BSs, and even w.r.t. independent placement of contents 
in the BSs.

\begin{figure*}[t]
\begin{minipage}[r]{0.5\linewidth}
\subfigure{
\includegraphics[width=\linewidth, height=7cm]{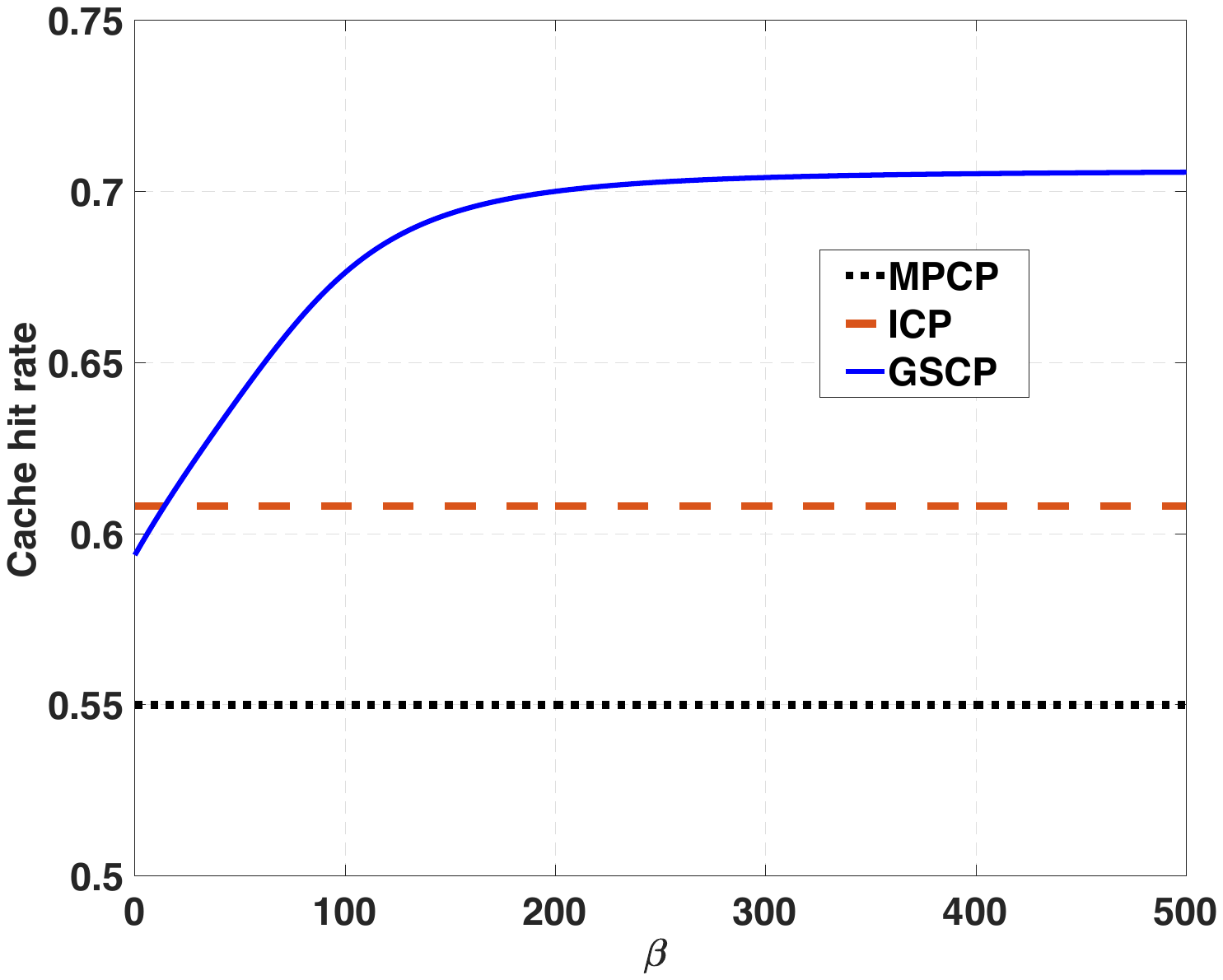}
\includegraphics[width=\linewidth, height=7cm]{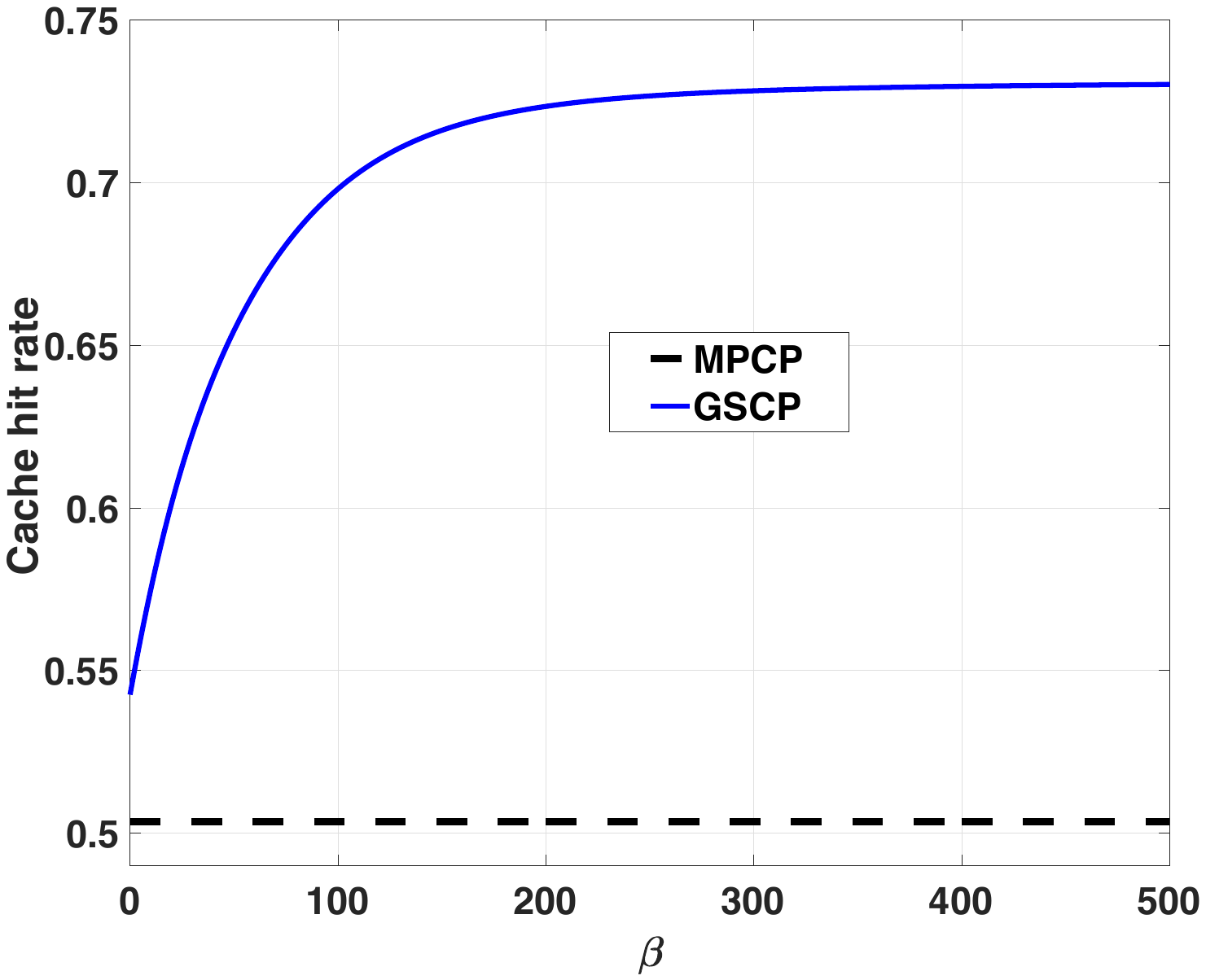}}
\end{minipage} \\ \hfill
\vspace{-15mm}
\caption{Comparison of Gibbs sampling based caching strategy, independent content placement strategy,  and most popular content 
placement strategy, for 
a network with six BSs, four possible contents, and storage capacity for two contents in each BS cache. Detailed system parameters 
can be found in Section~\ref{subsection:caching-two-dimensional-numerical-example}. The figure on the left is for scenario~$1$ and the figure on the right are for scenario~$2$ as described in Section~\ref{subsection:caching-two-dimensional-numerical-example}}
\label{fig:caching-plot-six-base-stations}
\vspace{-10mm}
\end{figure*}

\begin{figure*}[t]
\begin{minipage}[r]{0.5\linewidth}
\subfigure{
\includegraphics[width=\linewidth, height=5cm]{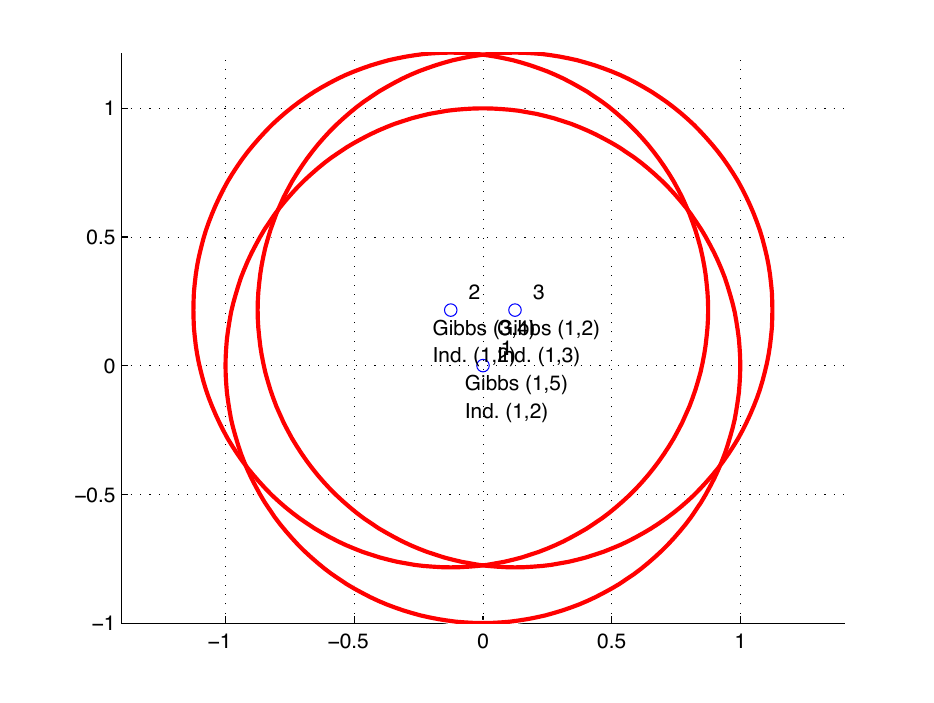}
\includegraphics[width=\linewidth, height=5cm]{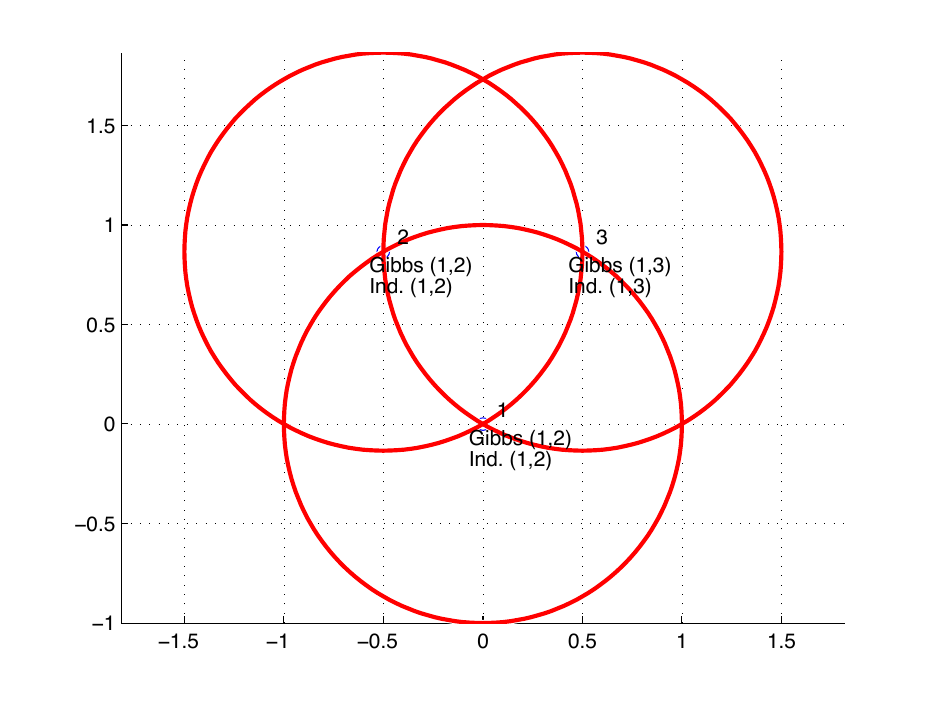}}
\end{minipage} \\ \hfill
\vspace{-5mm}
\begin{minipage}[r]{0.5\linewidth}
\subfigure{
\includegraphics[width=\linewidth, height=6cm]{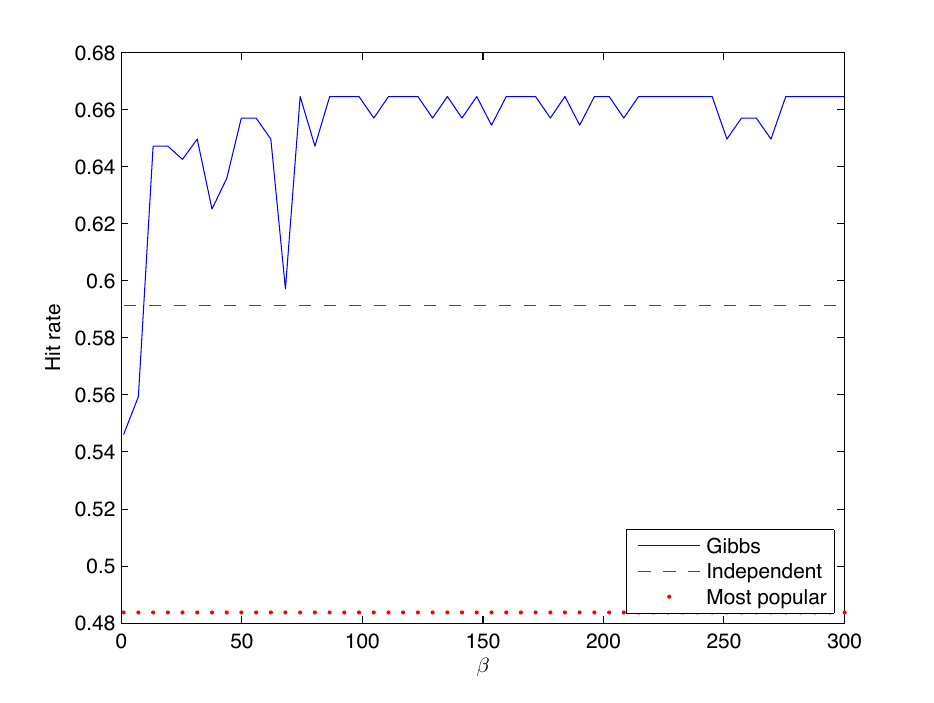}
\includegraphics[width=\linewidth, height=6cm]{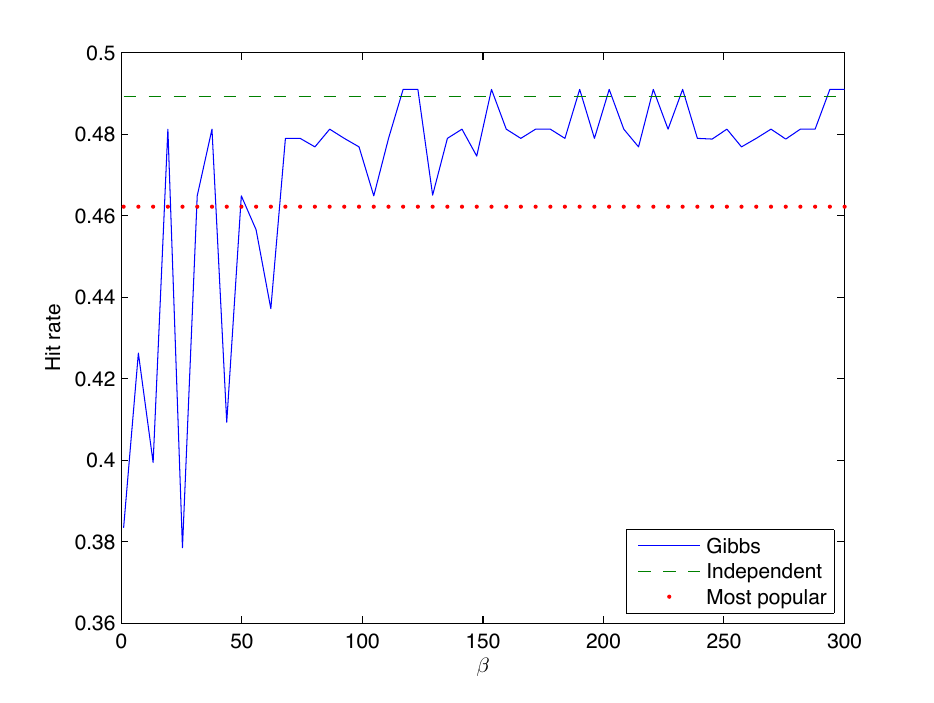}}
\end{minipage} \\ \hfill
\vspace{-15mm}
\caption{Comparison of Gibbs sampling  based caching with only $100$~iterations, against independent content placement strategy,  and most popular content placement strategy, for 
three BSs, five possible contents, and storage capacity of two contents per cache. Content popularities come from from Zipf distribution with parameter $\gamma=1.1$.  Details 
can be found in Section~\ref{subsection:gibbs-finite-iterations}. The top-left diagram shows more overlap among cells, whereas the top-right diagram shows less overlap. The diagrams at the bottom row correspond to the performance comparison among algorithms for these two cases.}
\label{fig:gibbs-plot-finite-iterations}
\vspace{-6mm}
\end{figure*}

\subsection{Effect of finite number of iterations, $\beta$, and cell overlap}\label{subsection:gibbs-finite-iterations}
In this subsection, we demonstrate the caching performance of Gibbs sampling with only a finite number of iterations. We consider two different cases: (i) three base stations on the plane, each with unit radius, more overlap among cells, and (ii) three base stations on the plane, each with unit radius, less overlap among cells. The set of contents are  $\mathcal{M}=\{1,2,3,4,5\}$ with their popularities coming from a Zipf distribution with parameter $\gamma=1.1$. Each cache can store at most two contents (i.e., $K=2$).

For these two cases, for various values of $\beta$, we simulated the Gibbs sampling algorithm (Algorithm~\ref{algorithm:virtual-cache-update-basic-gibbs-sampling}) for $100$~iterations, noted the configuration obtained after the $100$-th iteration, and computed the cache hit rates for these configurations via simulation. Next, we compared them against cache hit rates for most popular content placement and independent content placement schemes. The results are summarized in Figure~\ref{fig:gibbs-plot-finite-iterations}, where hit rates are computed per unit area of the entire window and not over the region covered by base stations alone. By the discussion provided in Section~\ref{subsection:asymptotic-performance}, we can expect that Gibbs sampling and independent content placement algorithms are both optimal if the cells become more overlapping. It is indeed seen in Figure~\ref{fig:gibbs-plot-finite-iterations} that the performances of Gibbs sampling and independent content placement algorithms are much better than most popular content placement, in case there is more overlapping among cells. It is also seen that the performance of Gibbs sampling tends to be better than independent content placement algorithm for large $\beta$. However, it is important to remember that we have only provided result for one sample path for each $\beta$; since we have taken only $100$~iterations for Gibbs sampling, the results will vary if another independent sample path is chosen for the Gibbs sampling algorithm. Hence, Figure~\ref{fig:gibbs-plot-finite-iterations} only demonstrates the {\em potential} performance improvement by Gibbs sampling over finite time; on the other hand, Section~\ref{subsection:caching-two-dimensional-numerical-example} demonstrates that Gibbs sampling  asymptotically achieves higher hit rate than independent content placement strategy  and most popular content placement strategy. 

\begin{figure*}[t]
\includegraphics[width=0.5 \linewidth, height=7cm]{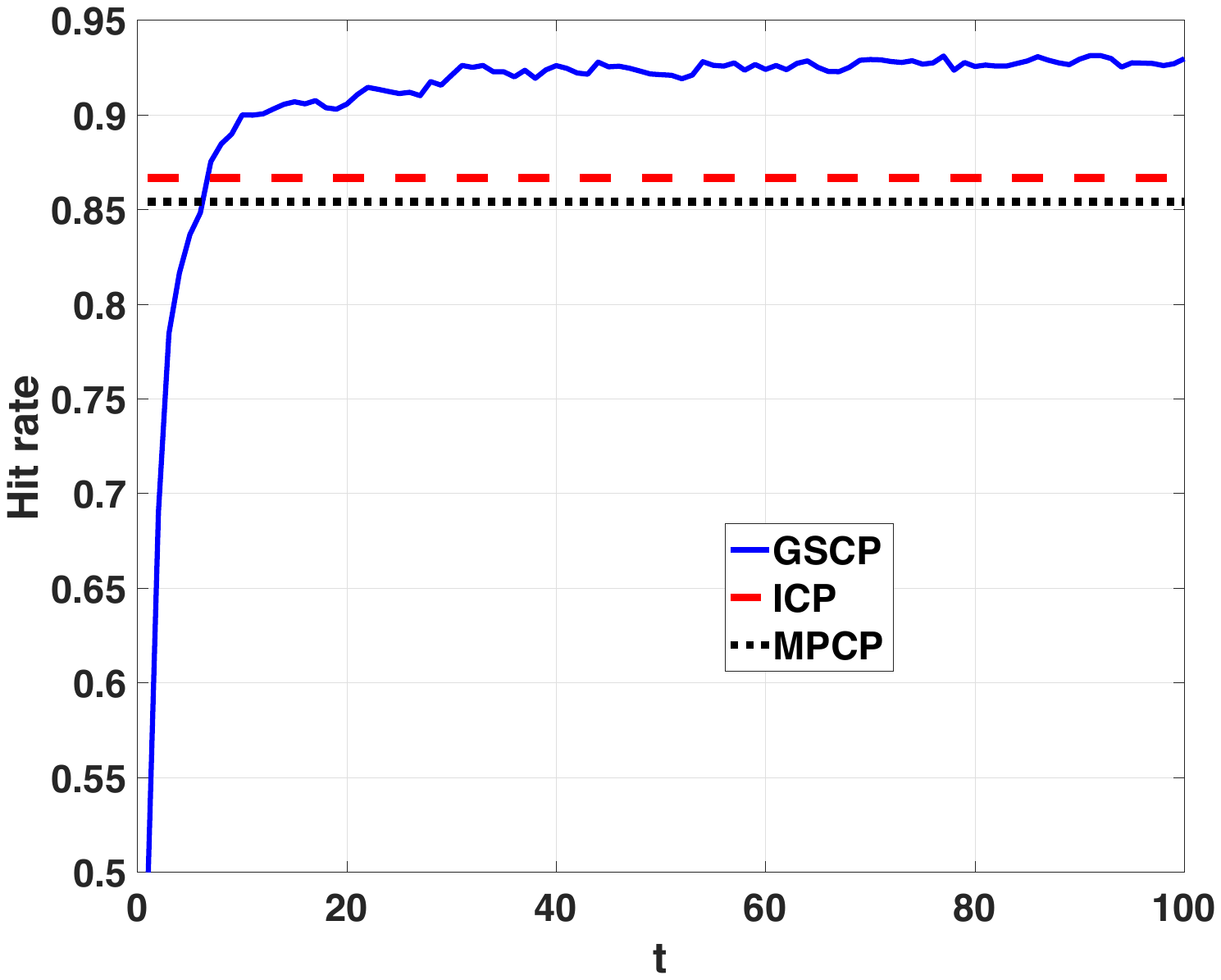}
\includegraphics[width=0.5 \linewidth, height=7cm]{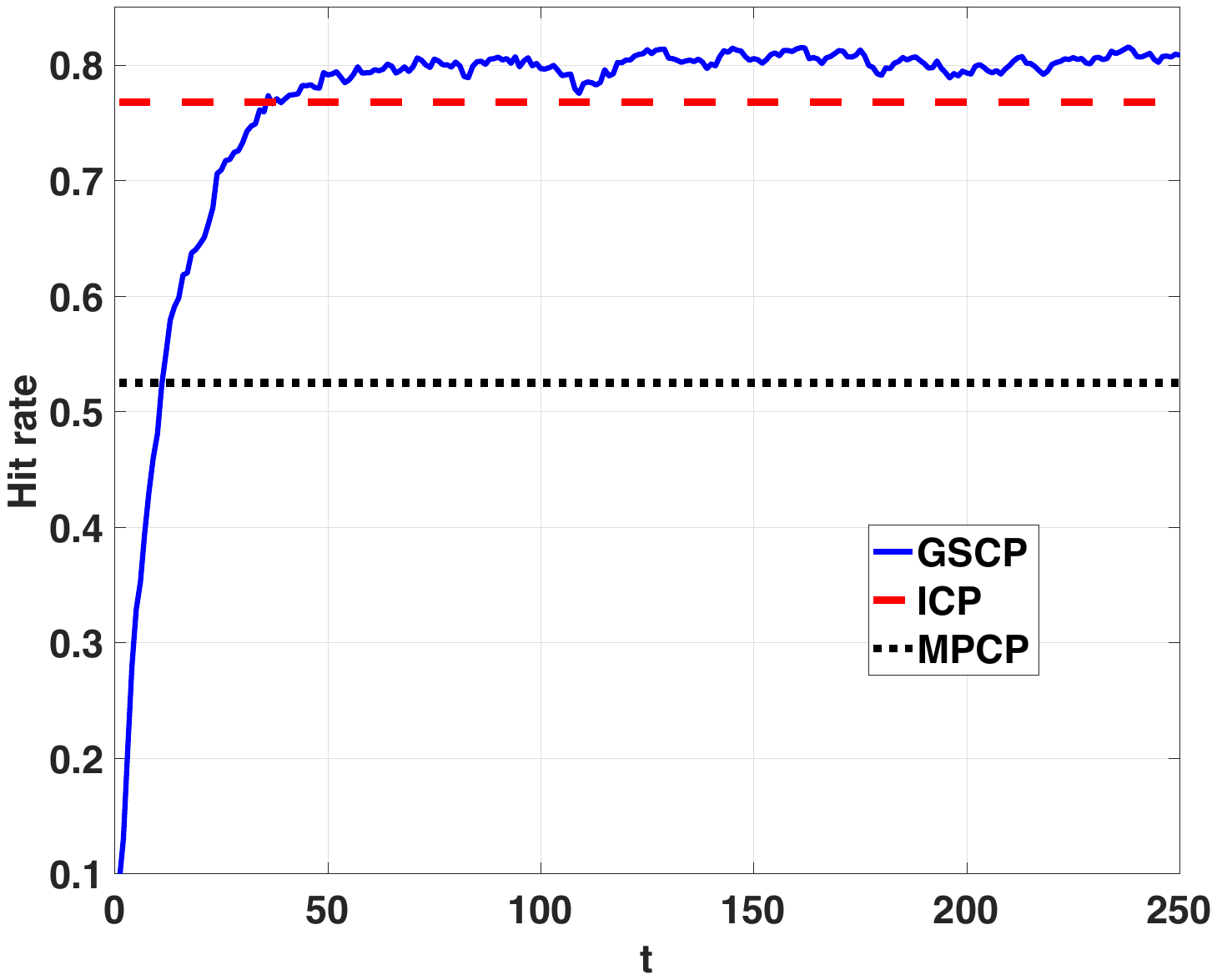}
\caption{Demonstration of convergence speed of GSCP under $\beta=150$, and performance improvement over ICP and MPCP.  Left: $N=5$, $M=5$, $K=2$, Zipf popularity distribution with parameter $\gamma=2$, averaged over $10$ sample paths. Right: $N=20$, $M=14$, $K=5$, Zipf popularity distribution with parameter $\gamma=0.5$, averaged over $4$ sample paths.  Details are provided in Section~\ref{subsection:mixing-time-given-beta}.}
\label{fig:mixing-time-plots}
\vspace{0mm}
\end{figure*}

\begin{figure*}[t]
\includegraphics[width=0.5 \linewidth, height=7cm]{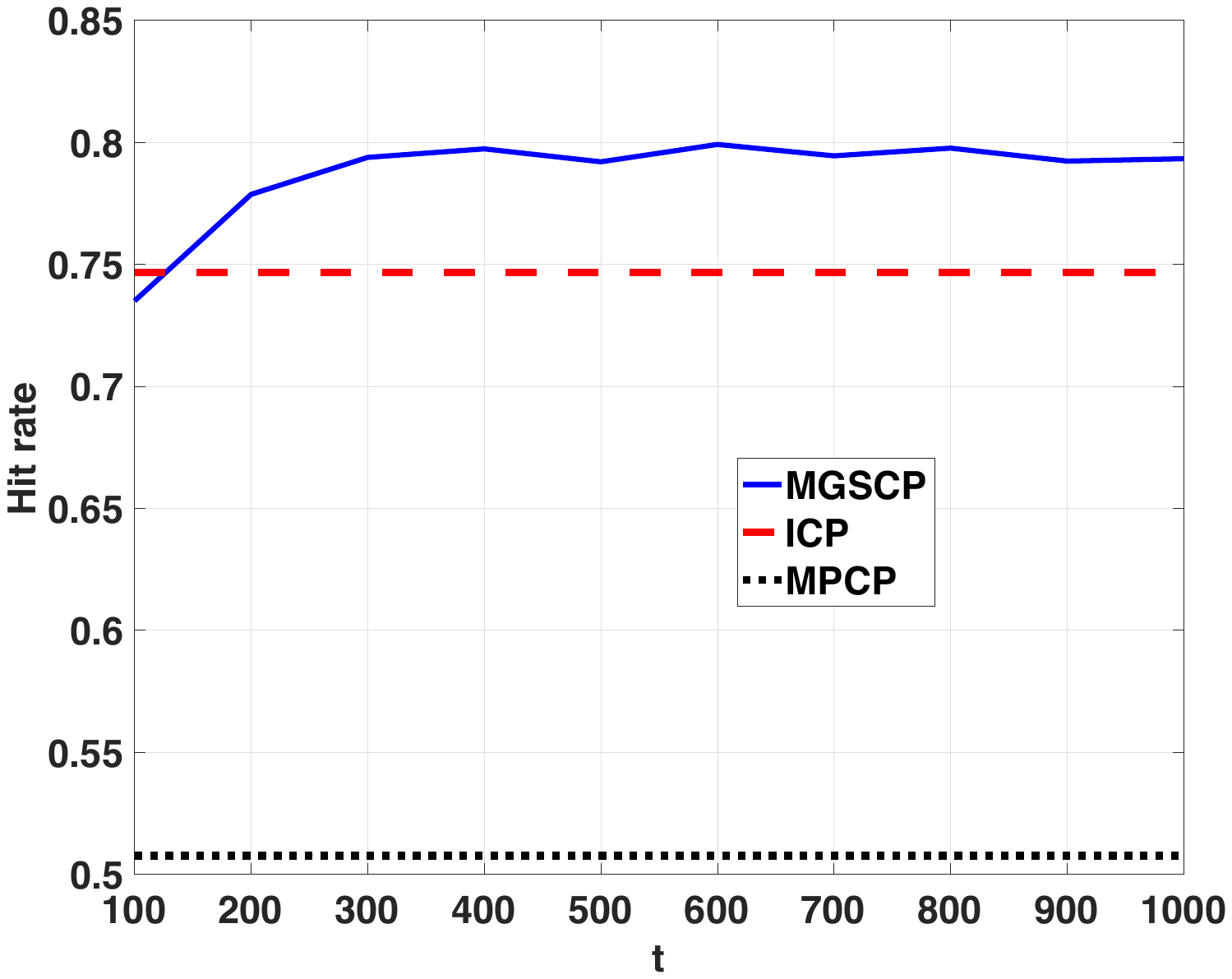}
\includegraphics[width=0.5\linewidth, height=7cm]{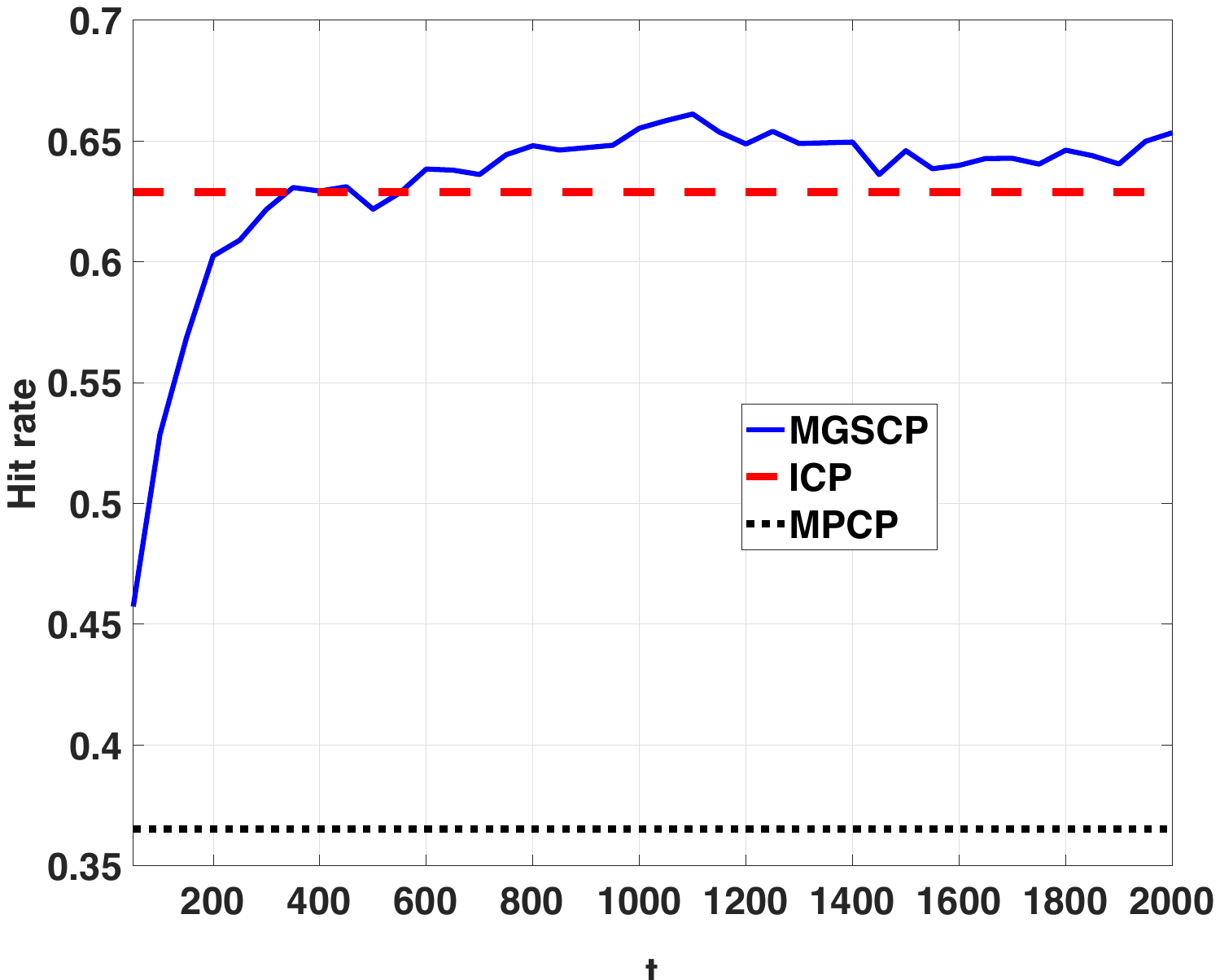}
\caption{Demonstration of convergence speed of modified GSCP (MGSCP), and performance improvement over ICP and MPCP.  Left: $N=20$, $M=50$, $K=15$, random popularity distribution, $\beta=1000$, averaged over $4$ sample paths. Right: $N=20$, $M=100$, $K=20$, random popularity distribution, $\beta=2000$,  single sample path. For MGSCP, all caches initially start with $K$ most popular contents. Details are provided in Section~\ref{subsection:mixing-time-given-beta}.}
\label{fig:large-scale-mixing-time-plots}
\vspace{0mm}
\end{figure*}

\subsection{Numerical example for mixing time and performance improvement of Gibbs sampling}\label{subsection:mixing-time-given-beta}
Now we demonstrate the speed of convergence of Gibbs sampling for fixed $\beta$. Location of $N$ base stations are generated independently with uniform distribution over the unit square, and the cell radius is assumed to be $\sqrt{\frac{5}{•\pi N}}$. Popularities of $M$ contents are generated either independently with uniform distribution, or they are assumed to follow Zipf distribution with parameter $\gamma$.  

{\bf Results for small system size:} For GSCP, each cache is assumed to be empty at $t=0$. The performance of GSCP for $\beta=150$, averaged over multiple independent sample paths, is compared against MPCP and ICP for various values of $N$, $M$ and the cache size $K$; for GSCP, at each $t$, hit rate for the current cache configuration is considered. Cache hit rate under GSCP is plotted against $t$ in Figure~\ref{fig:mixing-time-plots}. The results show that, GSCP outperforms MPCP and ICP significantly and reaches stationary distribution for even $t \leq 50$ if $N =5$; for $N=20$, the stationary distribution is nearly achieved starting from $t=100$. Of course, the convergence will be slower if $N$, $M$ and $K$ are increased further; for large values of $N$, $M$ and $K$, one can  simply use GSCP with only highly popular contents (for example, most popular contents whose collective popularity is $0.95$ or above). 

{\bf Results for large system size:} As discussed in Remark~\ref{remark:complexity-issue-in-basic-Gibbs}, the $O({{M} \choose {K}})$ computations per iteration in Gibbs sampling (using \eqref{eqn:third-expression-for-conditional-probability}) can be prohibitive for GSCP to be applied to a large scale system. We  alleviate this problem by proposing a simple modified GSCP algorithm (which we call MGSCP) where, at each iteration, only one randomly selected content is removed from a randomly selected BS, and then it is replaced by one content (absent in the cache after the removal) randomly via Gibbs sampling; thus, the denominator in \eqref{eqn:third-expression-for-conditional-probability} is replaced by a summation over all $(M-K+1)$ configurations that can possibly result via this replacement operation. Clearly this requires only $O(M-K+1)$ computations per iteration of Gibbs sampling and hence is easily implementable. This might reduce the convergence speed, but that can be compensated if one runs this iteration multiple times between two successive discrete time instants. However, here we assume that this update is done only once at each $t$. To reduce computation, we compute the hit rate only when $t$ is an integer multiple of either $50$ or $100$. Figure~\ref{fig:large-scale-mixing-time-plots} demonstrates that the MGSCP algorithm may take at most a few hundred iterations before it starts outperforming  ICP, and the convergence to steady state distribution is also clear from the plots; a few hundred iterations is not big for this large scale system (with $N=20$ and $M=50$ or $100$), especially keeping in mind that multiple iterations can be performed in practice between two successive decision instants. Thus, MGSCP provides a fast, distributed, optimal algorithm for content placement in a large system.

\vspace{-5mm}

\section{Conclusion}\label{section:conclusion}
In this paper, we have provided  
algorithms for cache content update in a cellular network, motivated by Gibbs sampling techniques. 
The algorithms were shown to converge asymptotically to the optimal content placement in the caches. It turns out that the computation and 
communication cost is affordable for practical cellular network base stations.

While the current paper solves an important problem, there are still possibilities for numerous interesting extensions: 
(i) We assumed uniform download cost from the backhaul network for all base stations. However, this is not in general true. Depending on the backhaul 
architecture, backhaul link capacities and congestion scenario, it might be more desirable to avoid download from some specific base stations. 
Even different base stations might have different link capacities, and in practice, this will result in queueing delay for the download process. 
Contents might be of various classes, and hence may not have fixed size. 
Hence, a combined formulation of cache update and backhaul network state evolution will be necessary. (ii) Different cells might witness different 
content popularities, but this has not been addressed in the current paper. (iii) Once a content becomes irrelevant (e.g., a news video), it has to be removed 
completely from all caches; one needs to develop techniques to detect when to remove a content from all caches. 
(iv) Providing convergence rate guarantees when the inverse temperature is increasing and when arrival rates and cell topology are learnt over time, 
is a very challenging problem. 
We leave these issues for future research endeavours on this topic.

{\small
\bibliographystyle{unsrt}
\bibliography{arpan-techreport}
}

\vspace{-20mm}

\begin{IEEEbiography}[{\includegraphics[width=1in,height=1in,clip,keepaspectratio]{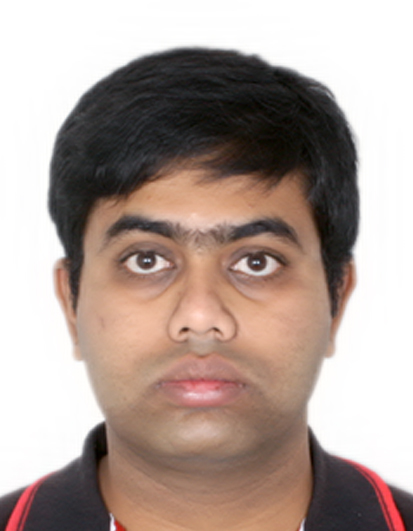}}]{Arpan 
Chattopadhyay} obtained his B.E. in Electronics and Telecommunication Engineering from Jadavpur University, 
Kolkata, India in the year 2008, and M.E. and Ph.D in Telecommunication Engineering from Indian Institute of Science, 
Bangalore, India in the year 2010 and 2015, respectively. He is currently working in the Ming Hsieh Department of Electrical Engineering, University of Southern California, Los Angeles  as a postdoctoral researcher. Previously he worked as a postdoc in INRIA/ENS Paris. 
His research interests include wireless networks, cyber-physical systems, machine learning and control.
    \end{IEEEbiography}

\vspace{-20mm}

\begin{IEEEbiography}[{\includegraphics[width=1in,height=1in,clip,keepaspectratio]{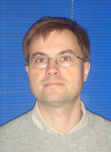}}]
{Bartlomiej Blaszczyszyn} received his PhD degree and Habilitation qualification
in applied mathematics from University of Wroclaw (Poland) in 1995
and 2008, respectively. He is now a Senior Researcher at Inria (France), and
a member of the Computer Science Department of Ecole Normale Superieure
in Paris. His professional interests are in applied probability, in particular in
stochastic modeling and performance evaluation of communication networks.
He coauthored several publications on this subject in major international
journals and conferences, as well as a two-volume book on {\em Stochastic
Geometry and Wireless Networks} NoW Publishers, jointly with F. Baccelli.
    \end{IEEEbiography}

\vspace{-20mm}

\begin{IEEEbiography}[{\includegraphics[width=1in,height=1in,clip,keepaspectratio]{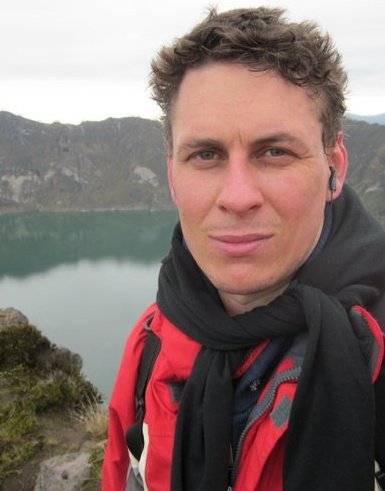}}]{H. Paul Keeler} graduated in physics and applied mathematics in 2006 from Griffith University. He received a Ph.D. in applied mathematics in 2010 from the University of Melbourne. After stints as a consultant in market research and banking and a guest researcher at the University of Zaragoza, Spain, he worked for two years as a researcher at Inria and École Normale Supérieure, Paris, where he was partly funded by Orange Labs. Currently he is a researcher at Weierstrass Institute (or WIAS), Berlin. His research interests lie in applied probability, numerical and asymptotic methods, communication networks.

    \end{IEEEbiography}

%
%

\appendices

\vspace{-5mm}

\section{Definition of weak and strong ergodicity}
\label{appendix:weak-and-strong-ergodicity}
Let us consider a discrete-time inhomogeneous Markov chain $\{X(t)\}_{t \geq 0}$ whose transition probability matrix (t.p.m.) between 
$t=m$ and $t=m+n$ is given by $P(m;n)$. Let $\mathcal{D}$ be the collection of all possible distributions 
(each element in $\mathcal{D}$ is assumed to be a row vector) on the state space. 
Then $\{X(t)\}_{t \geq 0}$ is called weakly ergodic if, for all $m \geq 0$,  
$$\lim_{n \uparrow \infty} \sup_{\mu,\nu \in \mathcal{D}}  d_V (\mu P(m;n) , \nu P(m;n) ) =0, $$
where $d_V(\cdot,\cdot)$ is the total variation distance between two distributions. 

$\{X(t)\}_{t \geq 0}$ is called strongly ergodic if there exists $\pi \in \mathcal{D}$ such that, for all $m \geq 0$, 
$$\lim_{n \uparrow \infty} \sup_{\mu \in \mathcal{D}}  d_V (\mu^{T} P(m;n) , \pi ) =0.$$

\vspace{-5mm}

\section{Proof of Theorem~\ref{theorem:asymptotic-optimality-of-real-cache-update}}
\label{appendix:real-cache-update-basic-gibbs-proofs}

Fix a small $\epsilon>0$. Under configuration $B$ of the {\em virtual caches}, 
let us denote the total time $T_B$ (a generic random variable) 
 taken by the arrival process so that, for  all possible pairs $(i,j) \in \{1,2,\cdots,M\} \times \{1,2,\cdots,N\}$ 
 there is at least one request for content $i$ to BS~$j$ if virtual configuration~$B$ 
 suggests placing content~$i$ at BS~$j$; clearly $\mathbf{E}(T_B)<\infty$, since we have made 
 Assumption~\ref{assumption:each-cell-has-a-region-covered-only-by-itself}.   
Let us consider $l \in \mathbb{Z}_{+}$ large enough such that: 
(i) $\sum_{B \in \mathcal{B}} |\mathbf{P}(V(t)=B)-\pi_{\beta}(B)|<\epsilon$ for all integer $t \geq S_{l-1}$, 
(ii) $\mathbf{P}(T_B>\epsilon T_{l+1}) < \epsilon$. 

Now, 

\footnotesize
\begin{eqnarray}
 && \frac{  \int_{S_l}^{S_{l+1}} \mathbf{P}(R(\tau)=B) d \tau  }{T_{l+1}} \nonumber\\
 & \geq & \frac{ \mathbf{E} \int_{\min\{S_l+T_B^{'},S_{l+1} \} }^{S_{l+1}} \mathbf{P}(R(\tau)=B) d \tau }{T_{l+1}} \nonumber\\
 & \geq & \frac{  \mathbf{P}(T_B^{'} \leq \epsilon T_{l+1})  \mathbf{E} \bigg(  \int_{\min\{S_l+T_B^{'},S_{l+1} \} }^{S_{l+1}} \mathbf{P}(R(\tau)=B) d \tau  \bigg| T_B^{'} \leq \epsilon T_{l+1}   \bigg)}{T_{l+1}  } \nonumber\\
  & \geq & \frac{ (1-\epsilon) \int_{S_l+\epsilon T_{l+1} }^{S_{l+1}} \mathbf{P}(R(\tau)=B| T_B^{'} \leq \epsilon T_{l+1}) d \tau }{T_{l+1}} \nonumber\\
  & = & \frac{ (1-\epsilon) (T_{l+1}-\epsilon T_{l+1}) \mathbf{P}(V(S_l-)=B) }{T_{l+1}} \nonumber\\
    & \geq & (1-\epsilon)^2  (\pi_{\beta}(B)-\epsilon), \nonumber\\
  \end{eqnarray}
  \normalsize
  
\noindent  where $T_B^{'}$ has the same distribution as $T_B$. The equality step follows from the fact that for $\tau > S_l+ \epsilon T_{l+1}$, we have $\mathbf{P}(R(\tau)=B| T_B^{'} \leq \epsilon T_{l+1}) =V(S_l-)$, since all real caches are updated to $V(S_l-)$ within  $\tau \leq S_l+T_B^{'}$.
  
  Hence, 
  \begin{eqnarray*}
 && \liminf_{T \rightarrow \infty} \frac{  \int_0^T \mathbf{P}(R(\tau)=B) d \tau  }{T} \\
 &=& \liminf_{T \rightarrow \infty} \frac{  \int_{S_l}^T \mathbf{P}(R(\tau)=B) d \tau  }{T-S_l} \\
 &\geq & (1-\epsilon)^2  (\pi_{\beta}(B)-\epsilon).
 \end{eqnarray*}

Since $\epsilon>0$ is arbitrarily small, we have: 
$$\liminf_{T \rightarrow \infty} \frac{  \int_0^T \mathbf{P}(R(\tau)=B) d \tau  }{T} \geq \pi_{\beta}(B).$$ 

But, by Fatou's lemma (see \cite[Chapter~$4$]{royden-fitzpatrick10real-analysis}), 
\begin{eqnarray*}
&& \sum_{B \in \mathcal{B}}\liminf_{T \rightarrow \infty} \frac{  \int_{0}^{T} \mathbf{P}(R(\tau)=B) d \tau  }{T} \\
& \leq & \liminf_{T \rightarrow \infty} \sum_{B \in \mathcal{B}} \frac{  \int_{0}^{T} \mathbf{P}(R(\tau)=B) d \tau  }{T}=1 
\end{eqnarray*}
and 
$\sum_{B \in \mathcal{B}} \pi_{\beta}(B)=1$. Hence, we must have 
$\liminf_{T \rightarrow \infty} \frac{  \int_{0}^{T} \mathbf{P}(R(\tau)=B) d \tau  }{T} = \pi_{\beta}(B)$ for all $B \in \mathcal{B}$. 

On the other hand, 
$$ \limsup_{T \rightarrow \infty} \frac{  \int_0^T \mathbf{P}(R(\tau)=B) d \tau  }{T} 
  = \limsup_{l \rightarrow \infty} \frac{  \int_{S_l}^{S_{l+1}} \mathbf{P}(R(\tau)=B) d \tau  }{T_{l+1}} .$$

Now, 

\small
\begin{eqnarray*}
 &&   \frac{  \int_{S_l}^{S_{l+1}} \mathbf{P}(R(\tau)=B) d \tau  }{T_{l+1}} \\
     & \leq &  \frac{  \int_{ S_l+ \epsilon T_{l+1} }^{S_{l+1}} \mathbf{P}(R(\tau)=B ) d \tau +  \epsilon T_{l+1}  }{T_{l+1}} \\
    & \leq &  \frac{  \int_{ S_l+ \epsilon T_{l+1} }^{S_{l+1}} \mathbf{P}(R(\tau)=B | T_B^{'} \leq \epsilon T_{l+1}) d \tau + 2 \epsilon T_{l+1}  }{T_{l+1}} \\
    & \leq &  \frac{  \int_{ S_l+ \epsilon T_{l+1} }^{S_{l+1}} \mathbf{P}(V(S_l-)=B) d \tau + 2 \epsilon T_{l+1}  }{T_{l+1}} \\
    &\leq & \pi_{\beta} (B) + 3 \epsilon ,
    \end{eqnarray*}
\normalsize

\noindent
where the second inequality follows from the fact that for $\tau \in [S_l+ \epsilon T_{l+1}, S_{l+1})$:

\small
\begin{eqnarray*}
 && \mathbf{P}(R(\tau)=B ) \\
 &\leq &  \mathbf{P}(R(\tau)=B | T_B^{'} \leq \epsilon T_{l+1}) \\
 && + \mathbf{P}(T_B^{'} > \epsilon T_{l+1}) \mathbf{P}(R(\tau)=B | T_B^{'} > \epsilon T_{l+1}) \\
 & \leq & \mathbf{P}(R(\tau)=B | T_B^{'} \leq \epsilon T_{l+1}) + \epsilon .
\end{eqnarray*}
\normalsize

Since $\epsilon>0$ is arbitrarily small, we can say that:  
$\limsup_{T \rightarrow \infty} \frac{  \int_0^T \mathbf{P}(R(\tau)=B) d \tau  }{T} \leq \pi_{\beta}(B).$

Hence, $\lim_{T \rightarrow \infty} \frac{  \int_0^T \mathbf{P}(R(\tau)=B) d \tau  }{T} = \pi_{\beta}(B)$.

\vspace{-5mm}

\section{Proof of Theorem~\ref{theorem:strong-ergodicity-varying-inverse-temperature}}\label{appendix:proof-of-strong-ergodicity-for-varying-inverse-temperature}
In this proof, we will use the notion of weak and strong ergodicity of time-inhomogeneous Markov chains from 
\cite[Chapter~$6$, Section~$8$]{breamud99gibbs-sampling}), which is provided in Appendix~\ref{appendix:weak-and-strong-ergodicity}.

Fix  $k =0$. 
We will first show that the Markov chain $\{V(t)\}_{t \geq 0}$ in weakly ergodic.

Let us consider the {\em transition probability matrix (t.p.m.)} $Q_l$ for the inhomogeneous Markov 
chain $\{Y(l)\}_{l \geq 0}$, where $Y(l):=V(lN)$.  
Then, the Dobrushin's ergodic coefficient $\delta(Q_l)$ is given by 
(see \cite[Chapter~$6$, Section~$7$]{breamud99gibbs-sampling} for definition) 
$\delta(Q_l)=1- \inf_{B^{'},B^{''} \in \mathcal{B}} \sum_{B \in \mathcal{B}} \min \{Q_l(B^{'},B),Q_l(B^{''},B) \}$. 
The Markov chain $\{V(t)\}_{t \geq 0}$ is weakly ergodic if $\sum_{l=1}^{\infty}(1-\delta(Q_l))=\infty$ (by 
\cite[Chapter~$6$, Theorem~$8.2$]{breamud99gibbs-sampling}).

Now, with positive probability, virtual caches in all nodes are updated over a period of $N$ slots. Hence, any  $B \in \mathcal{B}$ 
can be reached over a period of 
$N$ slots, starting from any other $B^{'} \in \mathcal{B}$. Note that, once a base station $j_t$ 
is chosen  in Algorithm~\ref{algorithm:virtual-cache-update-basic-gibbs-sampling} at discrete time $t \in \{lN,lN+1,\cdots,lN+N-1\}$, the sampling probability for any set of contents in 
its virtual cache in a slot 
is lower bounded by $\frac{e^{-\beta_t \Delta}}{{{M}\choose{K}}} \geq \frac{e^{-\beta_{lN+N} \Delta}}{{{M}\choose{K}}}$, since $t <lN+N$. 
Hence, for independent sampling over $N$ slots, we will always have 
$Q_l(B^{'},B) \geq \bigg( \frac{e^{-\beta_{lN+N} \Delta}}{N{{M}\choose{K}}} \bigg)^N >0$ for all pairs $B^{'},B$. 
Hence, 

\small
\begin{eqnarray}
&&\sum_{l=0}^{\infty}(1-\delta(Q_l)) \nonumber\\
&=& \sum_{l=0}^{\infty} \inf_{B^{'},B^{''} \in \mathcal{B}} \sum_{B \in \mathcal{B}} \min \{Q_l(B^{'},B),Q_l(B^{''},B) \} \nonumber\\
& \geq & \sum_{l=0}^{\infty} \sum_{B \in \mathcal{B}} \bigg( \frac{e^{-\beta_0 \log(1+lN+N) \times \Delta}}{N{ {M}\choose{K} }} \bigg)^N \nonumber\\
& = & { {M}\choose{K} }^N \times  \frac{1}{(N{ {M}\choose{K} })^N} \sum_{l=0}^{\infty}  e^{- N \Delta \beta_0 \log(1+lN+N)}  \nonumber\\
& = &  \frac{1}{N^N} \sum_{l=0}^{\infty}  \frac{1}{  (1+lN+N)^{\beta_0 N \Delta}}  \nonumber\\
& \geq &  \frac{1}{N^N } \frac{1}{N}\sum_{t=N+1}^{\infty}  \frac{1}{  (1+t)^{\beta_0 N \Delta}}  \nonumber\\
& = & \infty.
 \end{eqnarray}
\normalsize
Here the last step follows from the fact that $\sum_{t=1}^{\infty} \frac{1}{t^a}$ diverges for $0 <a<1$. The second equality follows from the fact that there are ${{M}\choose{K}}^N$ possible configurations.

 Hence, the Markov chain  $\{V(t)\}_{t \geq 0}$ is   weakly ergodic.
 
 Now we will use \cite[Chapter~$6$, Theorem~$8.3$]{breamud99gibbs-sampling} to prove strong ergodicity of 
 $\{V(t)\}_{t \geq 0}$.

 Let us denote the t.p.m. of $\{V(t)\}_{t \geq 0}$ at a specific time $t=T$ 
 by $Q^{(T)}$ (a specific matrix). If the Markov chain $\{V(t)\}_{t \geq 0}$ is allowed to evolve up to infinite time 
 with {\em fixed} t.p.m. $Q^{(T)}$, then we will get stationary distribution $\pi_{\beta_T}(B)= \frac{e^{\beta_T h(B)}}{Z_{\beta_T}}$. 
 This satisfies Condition~$8.9$ of \cite[Chapter~$6$, Theorem~$8.3$]{breamud99gibbs-sampling}. 
 
 Now we will check Condition~$8.10$ of \cite[Chapter~$6$, Theorem~$8.3$]{breamud99gibbs-sampling}. 
 For any $B \in \arg \max_{B^{'} \in \mathcal{B}} h(B^{'})$,  it is easy to see that $\pi_{\beta_T}(B)$ increases with $T$ for 
 large $T$ (can be seen by considering derivative of $\pi_{\beta}(B)$ w.r.t. $\beta$). For all other configurations $B$, 
 $\pi_{\beta_T}(B)$ decreases with $T$ for large $T$. 
 Hence, $\sum_{T=0}^{\infty} \sum_{B \in \mathcal{B}} |\pi_{\beta_{T+1}}(B)-\pi_{\beta_T}(B)| < \infty$. In order to see this, let us assume without loss of generality that, $B \in \arg \max_{B^{'} \in \mathcal{B}} h(B^{'})$ so that $\pi_{\beta_T}(B)$ monotonically increases with $T$ for all $T \geq T^{'}$. But $0 \leq \pi_{\beta_T}(B) \leq 1$ for all $T$. Hence, $\{ \pi_{\beta_T}(B) \}_{T \geq 1}$ converges and $\sum_{T=T^{'}}^{\infty}  |\pi_{\beta_{T+1}}(B)-\pi_{\beta_T}(B)| =\sum_{T=T^{'}}^{\infty}  (\pi_{\beta_{T+1}}(B)-\pi_{\beta_T}(B))=\lim_{T \rightarrow \infty} \pi_{\beta_T}(B)- \pi_{\beta_{T^{'}}}(B)< \infty$. Similar claims can be made for all $B \in \mathcal{B}$. Hence, we can claim that $\sum_{T=0}^{\infty} \sum_{B \in \mathcal{B}} |\pi_{\beta_{T+1}}(B)-\pi_{\beta_T}(B)| < \infty$.
 
 Hence, by \cite[Chapter~$6$, Theorem~$8.3$]{breamud99gibbs-sampling}, $\{V(t)\}_{t \geq 0}$ is strongly ergodic.  
 The expression for the limiting distribution is straightforward to derive.

\section{Proof of Theorem~\ref{theorem:convergence-virtual-cache-update-learning}}\label{appendix:proof-of-strong-ergodicity-for-varying-inverse-temperature-learning}
Note that, at a given fixed time $t=T$, given the instantaneous value of estimates, the instantaneous transition probability matrix for 
$\{V(t)\}_{t \geq 0}$ Markov chain, 
$Q^{(T)}$, will have a stationary probability distribution. Also, if we assume that there exists exactly one configuration in the set 
$\arg \max_{B \in \mathcal{B}}h(B)$, then we can say that $\lim_{T \rightarrow \infty} |Q^{(T)}-Q^*|=0$,  where $Q^*$ has a stationary 
distribution which assigns probability $1$ on $\arg \max_{B \in \mathcal{B}}h(B)$, and $Q^*$ is  ergodic. Hence, by 
\cite[Chapter~$6$, Theorem~$8.5$]{breamud99gibbs-sampling}, the Markov chain  $\{V(t)\}_{t \geq 0}$ is 
strongly ergodic.

\renewcommand{\thesubsection}{\Alph{subsection}}

\end{document}